\documentclass[preprint,12pt]{elsarticle}
\usepackage[T1]{fontenc}
\usepackage{todonotes}
\setuptodonotes{fancyline}
\usepackage{microtype}
\usepackage[font=small, labelfont=bf]{caption}
\usepackage[font=small, labelfont=small]{subcaption}
\usepackage{xspace,amssymb}

\usepackage{amssymb}
\usepackage{xcolor}
\usepackage{amsmath}
\usepackage{amsthm}
\definecolor{defblue}{rgb}{0.121,0.47,0.705}
\definecolor{linkblue}{rgb}{0.098,0.098,0.5392}
\definecolor{blue}{rgb}{0.098,0.098,0.5392}
\let\emph\relax
\DeclareTextFontCommand{\emph}{\color{defblue}\em}
\usepackage[colorlinks=true, 
linkcolor=linkblue, 
anchorcolor=linkblue,
citecolor=linkblue, 
filecolor=linkblue,
menucolor=linkblue,
urlcolor=linkblue, 
bookmarks=true, 
bookmarksopen=true,
bookmarksopenlevel=2, 
bookmarksnumbered=true, 
hyperindex=true,
plainpages=false, 
pdfpagelabels=true
]{hyperref}
\usepackage[capitalise,noabbrev,nameinlink]{cleveref}
\usepackage[inline]{enumitem}

\graphicspath{{figures/}}
\newcommand{\rot}{\textrm{rot}}
\newcommand{\yes}{\ensuremath{\mathsf{yes}}}
\newcommand{\inter}{\ensuremath{\mathsf{int}}}
\newcommand{\lab}{\ensuremath{\mathsf{label}}}
\newcommand{\pred}{\ensuremath{\mathsf{prev}}}
\newcommand{\need}{\ensuremath{\mathsf{next}}}
\newcommand{\intP}{\ensuremath{\mathsf{intPoint}}}
\newcommand{\len}{\ensuremath{\mathsf{weight}}}
\newcommand{\on}{\ensuremath{\mathsf{on}}}
\newcommand{\tru}{\ensuremath{\mathsf{true}}\xspace}
\newcommand{\fal}{\ensuremath{\mathsf{false}}\xspace}
\newcommand{\lef}{\ensuremath{\mathsf{left}}}
\newcommand{\dirG}{\ensuremath{G^\rightarrow}\xspace}
\newcommand{\maY}{\ensuremath{\mathsf{maxY}}}
\newcommand{\miY}{\ensuremath{\mathsf{minY}}}
\newcommand{\maX}{\ensuremath{\mathsf{maxX}}}
\newcommand{\miX}{\ensuremath{\mathsf{minX}}}
\newcommand{\nx}{\ensuremath{\mathsf{nearX}}}
\newcommand{\ny}{\ensuremath{\mathsf{nearY}}}
\newcommand{\rig}{\ensuremath{\mathsf{right}}}
\newcommand{\abo}{\ensuremath{\mathsf{above}}}
\newcommand{\bel}{\ensuremath{\mathsf{below}}}
\newcommand{\redu}{\ensuremath{\mathsf{red}}}
\newcommand{\ma}{\ensuremath{\mathsf{max}}}
\newcommand{\trun}{\ensuremath{\mathsf{trun}}}
\newcommand{\num}{\ensuremath{\mathsf{num}}}
\newcommand{\F}{\ensuremath{\mathsf{F}}\xspace}
\newcommand{\C}{\ensuremath{\mathsf{C}}\xspace}
\newcommand{\R}{\ensuremath{\mathsf{R}}\xspace}
\newcommand{\E}{\ensuremath{\mathsf{E}}\xspace}
\newcommand{\W}{\ensuremath{\mathsf{W}}\xspace}
\newcommand{\N}{\ensuremath{\mathsf{N}}\xspace}
\newcommand{\South}{\ensuremath{\mathsf{S}}\xspace}
\newcommand{\XS}{\ensuremath{\mathsf{S}}\xspace}
\newcommand{\XL}{\ensuremath{\mathsf{L}}\xspace}
\newcommand{\mycap}{\ensuremath{\mathsf{cap}}}
\newcommand{\RCR}{\ensuremath{\mathsf{RCR}}\xspace}
\newcommand{\RCCC}{\ensuremath{\mathsf{RCCC}}\xspace}
\newcommand{\RCCR}{\ensuremath{\mathsf{RCCR}}\xspace}
\newcommand{\RRRC}{\ensuremath{\mathsf{RRRC}}\xspace}
\newcommand{\CRC}{\ensuremath{\mathsf{CRC}}\xspace}
\newcommand{\CRRC}{\ensuremath{\mathsf{CRRC}}\xspace}
\newcommand{\CRRR}{\ensuremath{\mathsf{CRRR}}\xspace}
\newcommand{\CCCR}{\ensuremath{\mathsf{CCCR}}\xspace}

\crefname{figure}{Fig.}{Figs.}
\Crefname{figure}{Figure}{Figures}
\newtheorem{reduction}{Reduction Rule}
\newtheorem{counting}{Counting Rule}
\newtheorem{observation}{Observation}
\newtheorem{lemma}{Lemma}
\newtheorem{theorem}{Theorem}
\newtheorem{proposition}{Proposition}

\usepackage{lineno}

\journal{arXiv}

\begin{document}

\begin{frontmatter}

\title{Parameterized Approaches to Orthogonal Compaction\tnoteref{titleRef}}
\tnotetext[titleRef]{A preliminary version of this paper appeared in the proceedings of SOFSEM 2023.}

\author[1]{Walter Didimo\fnref{authorRef1}}
\ead{walter.didimo@unipg.it}
\fntext[authorRef1]{Work partially supported by Dept.\ Engineering, Università degli Studi di Perugia, grant RICBA21LG: Algoritmi, modelli e sistemi per la rappresentazione visuale di reti.}
\affiliation[1]{organization={Università degli Studi di Perugia}, city={Perugia}, country={Italy}}
             
\author[2]{Siddharth Gupta}
\ead{siddharthg@goa.bits-pilani.ac.in}
\affiliation[2]{organization={BITS Pilani, K K Birla Goa Campus}, city={Goa}, country={India}}

\author[3]{Philipp Kindermann}
\ead{kindermann@uni-trier.de}
\affiliation[3]{organization={Universität Trier}, city={Trier}, country={Germany}}

\author[1]{Giuseppe~Liotta}
\ead{giuseppe.liotta@unipg.it}

\author[4]{Alexander~Wolff}
\affiliation[4]{organization={Universität Würzburg}, city={Würzburg}, country={Germany}}

\author[5]{Meirav~Zehavi\fnref{authorRef6}}
\ead{meiravze@bgu.ac.il}
\fntext[authorRef6]{Work partially supported by European Research Council (ERC) grant PARAPATH.}
\affiliation[5]{organization={Ben-Gurion University of the Negev}, city={Beersheba}, country={Israel}}

\begin{abstract}
  Orthogonal graph drawings are used in applications such as UML
  diagrams, VLSI layout, cable plans, and metro maps.  We focus on
  drawing planar graphs and assume that we are given an
  \emph{orthogonal representation} that describes the desired shape,
  but not the exact coordinates of a drawing.  Our aim is to compute
  an orthogonal drawing on the grid that has minimum area among all
  grid drawings that adhere to the given orthogonal
  representation. \par %
  This problem is called orthogonal compaction (\textsc{OC}) and is
  known to be NP-hard, even for orthogonal representations of cycles
  [Evans et al.\ 2022].  We investigate the complexity of \textsc{OC}
  with respect to several parameters.  Among others, we show that
  \textsc{OC} is fixed-parameter tractable with respect to the most
  natural of these parameters, namely, the number of \emph{kitty
    corners} of the orthogonal representation: the presence of pairs
  of kitty corners in an orthogonal representation makes the
  \textsc{OC} problem hard. Informally speaking, a pair of kitty
  corners is a pair of reflex corners of a face that point at each
  other.  Accordingly, the number of kitty corners is the number of
  corners that are involved in some pair of kitty corners.
\end{abstract}

\begin{keyword}
  Orthogonal Graph Drawing \sep Orthogonal Representation \sep
  Compaction \sep Parameterized Complexity
\end{keyword}

\end{frontmatter}

\section{Introduction}

In a \emph{planar orthogonal drawing} of a planar graph $G$ each
vertex is mapped to a distinct point of the plane and each edge is
represented as a sequence of horizontal and vertical segments.  A
planar graph $G$ admits a planar orthogonal drawing if and only if it
has vertex-degree at most four. A \emph{planar orthogonal
  representation}~$H$ of~$G$ is an equivalence class of planar
orthogonal drawings of $G$ that have the same ``shape'', i.e., the
same planar embedding, the same ordered sequence of bends along the
edges, and the same vertex angles. A planar orthogonal drawing
belonging to the equivalence class $H$ is simply called a
\emph{drawing of $H$}. For example, \cref{fig:comp-1,fig:comp-2} are
drawings of the same orthogonal representation, while
\cref{fig:comp-3} is a drawing of the same graph with a different
shape.

Given a planar orthogonal representation $H$ of a connected planar
graph~$G$, the \emph{orthogonal compaction} problem (\textsc{OC} for
short) for~$H$ asks to compute a minimum-area drawing of~$H$. More
formally, it asks to assign integer coordinates to the vertices and to
the bends of~$H$ such that the area of the resulting planar orthogonal
drawing is minimum over all drawings of~$H$. The area of a drawing is
the area of the minimum bounding box that contains the drawing. For
example, the drawing in \cref{fig:comp-1} has area $7 \times 5 = 35$,
whereas the drawing in \cref{fig:comp-2} has area $7 \times 4 = 28$,
which is the minimum for that orthogonal representation.

The area of a graph layout is considered one of the most relevant
readability metrics in orthogonal graph drawing (see,
e.g.,~\cite{DBLP:books/ph/BattistaETT99,DBLP:conf/dagstuhl/1999dg}).
Compact grid drawings are desirable as they yield a good overview
without neglecting details.  For this reason, the \textsc{OC} problem
is widely investigated in the literature.  Bridgeman et
al.~\cite{DBLP:journals/comgeo/BridgemanBDLTV00} showed that
\textsc{OC} can be solved in linear time, in the number of vertices and bends, for a subclass of planar
orthogonal representations called \emph{turn-regular}. Informally
speaking, a face of a planar orthogonal representation~$H$ is
turn-regular if it does not contain any pair of so-called \emph{kitty
  corners}, i.e., a pair of reflex corners (turns of $270^\circ$) that
point to each other; a representation is turn-regular if all its faces
are turn-regular.  See \cref{fig:compaction} and refer to
\cref{se:basics} for a formal definition.  On the other hand,
Patrignani~\cite{DBLP:journals/comgeo/Patrignani01} proved that,
unfortunately, \textsc{OC} is NP-hard in general.  Evans et
al.~\cite{efkssw-mrpgas-CGTA22} showed that \textsc{OC} remains
NP-hard even for orthogonal representations of simple cycles.  Since
cycles have constant pathwidth (namely~2), this immediately shows that
we cannot expect an FPT (or even an XP) algorithm parameterized by
pathwidth alone unless P = NP.  The same holds for parametrizations
with respect to treewidth since the treewidth of a graph is upper
bounded by its pathwidth.

In related work, Bannister et al.~\cite{bes-ioc-JGAA12} showed that
several problems of compacting {\em not necessarily planar}
orthogonal graph drawings to use the minimum number of rows, area,
length of longest edge, or total edge length cannot be approximated
better than within a polynomial factor of optimal (if P$\,\ne\,$NP).
They also provided an FPT algorithm for testing whether a drawing can
be compacted to a number of rows $k$; the algorithm is parameterized
by the natural parameter $k$.
Note that their algorithm does not solve the planar case because the
algorithm is allowed to change the embedding.

\begin{figure}[tb]
  \begin{subfigure}[b]{.32\linewidth}
    \centering
    \includegraphics[page=2]{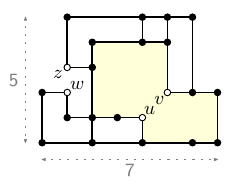}
    \caption{area = 35}
    \label{fig:comp-1}
  \end{subfigure}
  \hfill
  \begin{subfigure}[b]{.32\linewidth}
    \centering
    \includegraphics[page=4]{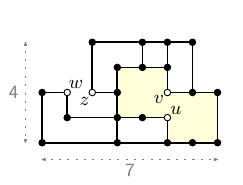}
    \caption{area = 28}
    \label{fig:comp-2}
  \end{subfigure}
  \hfill
  \begin{subfigure}[b]{.32\linewidth}
    \centering
    \includegraphics[page=6]{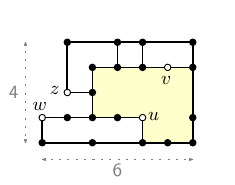}
    \caption{area = 24}
    \label{fig:comp-3}
  \end{subfigure}
  \caption{(a)~Drawing of a non-turn-regular orthogonal representation
    $H$; vertices $u$ and $v$ point to each other in the filled
    internal face, thus they represent a pair of kitty
    corners. Vertices $w$ and $z$ are a pair of kitty corners in the
    external face.  (b)~Another drawing of $H$ with minimum
    area. (c)~Minimum-area drawing of a turn-regular orthogonal
    representation of the same graph.}
  \label{fig:compaction}
\end{figure}

The research in this paper is motivated by the relevance of the
\textsc{OC} problem and by the growing interest in parameterized
approaches for NP-hard problems in graph
drawing~\cite{DBLP:journals/dagstuhl-reports/GanianMNZ21}. Recent
works on the subject include parameterized algorithms for book
embeddings and queue
layouts~\cite{DBLP:journals/jgaa/BannisterE18,DBLP:journals/jgaa/BhoreGMN20,DBLP:journals/jgaa/BhoreGMN22,DBLP:conf/compgeom/BinucciLGDMP19,kobayashi_et_al:LIPIcs:2018:8566},
upward planar
drawings~\cite{DBLP:conf/compgeom/BinucciLGDMP19,cdfgrs-paup-SoCG22},
orthogonal planarity testing and grid
recognition~\cite{DBLP:journals/jcss/GiacomoLM22,DBLP:conf/isaac/0002SZ21},
clustered planarity and hybrid
planarity~\cite{DBLP:conf/wg/LozzoEG018,LiRuTa-PCGPRCO-22,DBLP:journals/algorithmica/LozzoEGG21},
1-planar drawings~\cite{DBLP:journals/jgaa/BannisterCE18},
crossing
minimization~\cite{DBLP:conf/gd/BannisterES13,DBLP:journals/algorithmica/DujmovicFKLMNRRWW08,DBLP:journals/jda/DujmovicFK08},
and visual complexity
\cite{cdggkw-pcsn-GD23,cflrvw-cdgfl-JGAA23,bcghvw-bcong-SIDMA24}.

\paragraph{Contribution} Extending this line of research, we initiate
the study of the parameterized complexity of \textsc{OC} and
investigate several parameters:
\begin{itemize}
\item {\em Number of kitty corners.}  Given that \textsc{OC} can be
  solved efficiently for orthogonal representations without kitty
  corners, the number of kitty corners (that is, the number of corners
  involved in some pair of kitty corners) is a very natural parameter
  for \textsc{OC}.  We show that \textsc{OC} is fixed-parameter
  tractable (FPT) with respect to the number of kitty corners
  (\cref{th:fpt-kitty} in \cref{se:fpt-kitty}).
      
\item {\em Number of faces.} Since \textsc{OC} remains NP-hard for
  orthogonal representations of simple
  cycles~\cite{efkssw-mrpgas-CGTA22}, \textsc{OC} is para-NP-hard when
  parameterized by the number of faces. Hence, we cannot expect an FPT
  (or even an XP) algorithm in this parameter alone, unless
  P$\,=\,$NP. However, for orthogonal representations of simple cycles
  we show the existence of a polynomial kernel for \textsc{OC} when
  parameterized by the number of kitty corners
  (\cref{the:polynomial-kernel} in \cref{sec:kernel}).
    
\item {\em Maximum face-degree.} The maximum face-degree is the
  maximum number of vertices on the boundary of a face. Since both the
  NP-hardness reductions by Patrignani~\cite{DBLP:journals/comgeo/Patrignani01} and by Evans et al.~\cite{efkssw-mrpgas-CGTA22} require faces of linear size, it is interesting to know whether faces of constant size make the problem
  tractable.  We prove that this is not the case, i.e., \textsc{OC}
  remains NP-hard when parameterized by the maximum face degree
  (\cref{th:hard-face-degree} in \cref{se:hardness}).
    
\item {\em Height.} The \emph{height} of an orthogonal representation
  is the minimum number of distinct y-coordinates required to draw the
  representation.  Since a $w \times h$ grid has pathwidth at
  most~$h$, graphs with bounded height have bounded pathwidth, but the
  converse is generally not true~\cite{DBLP:journals/dcg/Biedl11}.  In
  fact, we show that \textsc{OC} admits an XP algorithm parameterized
  by the height of the given representation (see \cref{th:xp-height}
  in \cref{se:xp-height}).  In this context, we remark that a related
  problem has been considered by Chaplick et
  al.~\cite{cflrvw-dgflf-JoCG20}.
  Namely, given a planar graph~$G$, they define $\bar{\pi}^1_2(G)$ to
  be the minimum number of distinct y-coordinates required to draw the
  graph straight-line, and provide lower and upper bounds on this
  parameter. In their setting, however, the embedding of~$G$ is not
  fixed, and the graph is not necessarily drawn in the orthogonal
  drawing style.
\end{itemize}
We start with some basics in \cref{se:basics} and close with open
problems in \cref{se:conclusions}.

\section{Basic Definitions}
\label{se:basics}

Let $G=(V,E)$ be a connected planar graph of vertex-degree at most
four and let $\Gamma$ be a planar orthogonal drawing of $G$. We assume
that in $\Gamma$ all the vertices and bends have integer coordinates,
i.e., we assume that $\Gamma$ is an integer-coordinate grid
drawing.
The connected regions of the plane determined by $\Gamma$ are called
\emph{faces}; each face is described by the clockwise sequence of
vertices and edges along its boundary. Exactly one of these faces is
an infinite region and it is called the \emph{external face} (or
\emph{outer face}). The other faces of $\Gamma$ are the \emph{internal
  faces}. The \emph{planar embedding} of $\Gamma$ is the set of
(internal and external) faces determined by $\Gamma$.  Given a face
$f$ of $\Gamma$, the \emph{degree of $f$ } is the number of vertices
traversed while walking clockwise on its boundary, where each vertex
is counted the number of times it is traversed (if $G$ is biconnected,
each vertex of $f$ is traversed exactly once).

Two planar orthogonal drawings~$\Gamma_1$ and~$\Gamma_2$ of
$G$ are \emph{shape-equivalent} if: $(i)$~$\Gamma_1$ and~$\Gamma_2$
have the same planar embedding; $(ii)$~for each vertex $v \in V$, the
geometric angles at $v$ (formed by any two consecutive edges incident
on $v$) are the same in $\Gamma_1$ and $\Gamma_2$; $(iii)$~for each
edge $e=(u,v) \in E$ the sequence of left and right bends along $e$
while moving from $u$ to $v$ is the same in $\Gamma_1$ and
$\Gamma_2$. An \emph{orthogonal representation} $H$ of $G$ is a class
of shape-equivalent planar orthogonal drawings of $G$. It follows that
an orthogonal representation $H$ is completely described by a planar
embedding of $G$, by the values of the angles around each vertex (each
angle being a value in the set
$\{90^\circ,180^\circ,270^\circ,360^\circ\}$), and by the ordered
sequence of left and right bends along each edge $(u,v)$, moving from
$u$ to $v$; if we move from $v$ to $u$, then this sequence and the
direction (left/right) of each bend are reversed. If $\Gamma$ is a
planar orthogonal drawing in the class $H$, then we also say that
$\Gamma$ is a drawing of $H$. Without loss of generality, we also
assume that an orthogonal representation $H$ comes with a given
``orientation'', i.e., for each edge segment $\overline{pq}$ of $H$
(where $p$ and $q$ correspond to vertices or bends), we fix whether
$p$ lies to the left, to the right, above, or below~$q$.

\paragraph{Turn-regular orthogonal representations}
Let $H$ be a planar orthogonal representation. For the purpose of the
\textsc{OC} problem, and without loss of generality, we always assume
that each bend in $H$ is replaced by a degree-2 vertex, i.e., $H$ is a \emph{rectilinear representation} (each edge is either a horizontal or a vertical segment).  Let $f$ be a
face of~$H$ and assume that the
boundary of $f$ is traversed counterclockwise (resp. clockwise) if $f$
is internal (resp. external).  Let $u$ and $v$ be two reflex vertices
of $f$. Let $\rot(u,v)$ be the number of convex corners minus the
number of reflex corners encountered while traversing the boundary of
$f$ from $u$ (included) to $v$ (excluded); a reflex vertex of degree
one is counted like two reflex vertices. We say that $u$ and $v$ is a
pair of \emph{kitty corners} of~$f$ if $\rot(u,v)=2$ or
$\rot(v,u)=2$. A vertex is a kitty corner if it is part of a pair of
kitty corners.  A face $f$ of $H$ is \emph{turn-regular} if it does
not contain a pair of kitty corners.  The representation~$H$ is
\emph{turn-regular} if all faces are turn-regular.

\paragraph{Parameterized Complexity}

Let $\Pi$ be an NP-hard problem. In the framework of parameterized
complexity, each instance of $\Pi$ is associated with a \emph{parameter}~$k$. Here, the goal is to confine the combinatorial
explosion in the running time of an algorithm for $\Pi$ to depend only
on $k$. Formally, we say that $\Pi$ is \emph{fixed-parameter tractable (FPT)} if any instance $(I, k)$ of $\Pi$ is solvable in time
$f(k)\cdot |I|^{O(1)}$, where $f$ is an arbitrary computable function
of $k$. A weaker request is that for every fixed $k$, the problem
$\Pi$ would be solvable in polynomial time. Formally, we say that
$\Pi$ is \emph{slice-wise polynomial (XP)} if any instance $(I, k)$ of
$\Pi$ is solvable in time $f(k)\cdot |I|^{g(k)}$, where $f$ and $g$
are arbitrary computable functions of $k$.

A companion notion of fixed-parameter tractability is that of \emph{kernelization}. A \emph{compression algorithm} from a problem $\Pi$
to a problem $\Pi'$ is a polynomial-time algorithm that transforms an
arbitrary instance of $\Pi$ to an equivalent instance of $\Pi'$ whose
size is bounded by some computable function $g$ of the parameter $k$ of
the original instance. When $\Pi'=\Pi$, then the compression algorithm
is called a kernelization algorithm, and the resulting instance is
called a \emph{kernel}. Further, we say that $\Pi$ admits a compression
(or a kernel) of size $g(k)$ where $k$ is the parameter. 
It is well known that a problem is fixed-parameter tractable (FPT) by
a parameter $k$ if and only if it admits a kernelization algorithm for
the parameter $k$.
If $g(k)$ is the size of a kernel and $g$ is a polynomial function,
then the compression (or kernel) is called a \emph{polynomial
  compression} (or \emph{polynomial kernel}).  For more information on
parameterized complexity, we refer to books such as
\cite{DBLP:books/sp/CyganFKLMPPS15,downey2013fundamentals,fomin2019kernelization}.

\section{Number of Kitty Corners: An FPT Algorithm}
\label{se:fpt-kitty}

Turn-regular orthogonal representations can be compacted optimally in
linear time~\cite{DBLP:journals/comgeo/BridgemanBDLTV00}.  
The result in~\cite{DBLP:journals/comgeo/BridgemanBDLTV00} exploits an interesting relationship between planar orthogonal representations and upward planar embeddings. In the following, we first introduce the necessary terminology about upward planarity, then we recall the result of~\cite{DBLP:journals/comgeo/BridgemanBDLTV00}, and finally describe our FPT algorithm.

\paragraph{Upward planar embeddings and saturators}

Let $D=(V,E)$ be a plane DAG, i.e., an acyclic digraph with a given
planar embedding.  An \emph{upward planar drawing} $\Gamma$ of $D$ is
an embedding-preserving drawing of $D$ where each vertex $v$ is mapped
to a distinct point of the plane and each edge is drawn as a simple
Jordan arc monotonically increasing in the upward direction.  Such a
drawing exists if and only if $D$ is the spanning subgraph of a plane
\emph{$st$-graph}, i.e., a plane digraph with a unique source $s$ and
a unique sink $t$, which are both on the external
face~\cite{DBLP:journals/tcs/BattistaT88}.  Let $S$ be the set of
sources of~$D$, $T$ be the set of sinks, and
$I = V \setminus (S \cup T)$.  $D$~is~\emph{bimodal}~if, for every
vertex $v \in I$, the outgoing edges (and hence the incoming edges)
of~$v$ are consecutive in the clockwise order around~$v$.  If an
upward planar drawing $\Gamma$ of $D$ exists, then $D$ is necessarily
bimodal and $\Gamma$ uniquely defines the left-to-right orderings of
the outgoing and incoming edges of each vertex.  This set of orderings
(for all vertices of $D$) is an \emph{upward planar embedding} of~$D$,
and is regarded as an equivalence class of upward planar drawings of~$D$.
A plane DAG with a given upward planar embedding is an \emph{upward
  plane DAG}.

\begin{figure}[tb]
  \begin{subfigure}[b]{0.5\linewidth}
    \centering \includegraphics[page=1]{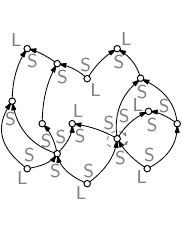}
    \caption{}
    \label{fig:updag-1}
  \end{subfigure}
  \hfill
  \begin{subfigure}[b]{0.5\linewidth}
    \centering \includegraphics[page=2]{upward-plane-DAG}
    \caption{}
    \label{fig:updag-2}
  \end{subfigure}
  \caption{(a) An upward plane DAG $D$ and the corresponding upward
    labeling. (b) A plane $st$-graph obtained by augmenting $D$ with a
    complete saturator (dotted edges).}
  \label{fig:updag}
\end{figure}

Let $e_1$ and $e_2$ be two consecutive edges on the boundary of a face
$f$ of a bimodal plane digraph $D$, and let $v$ be their common
vertex.  Vertex~$v$ is a \emph{source switch of~$f$} (resp. a
\emph{sink switch of~$f$}) if both $e_1$ and $e_2$ are outgoing edges
(resp. incoming edges) of $v$.  Note that, for each face $f$, the
number $n_f$ of source switches of~$f$ equals the number of sink
switches of~$f$. The \emph{capacity of $f$} is the function
$\mycap(f)=n_f-1$ if $f$ is an internal face and $\mycap(f)=n_f+1$ if
$f$ is the external face.  If $\Gamma$ is an upward planar drawing
of~$D$, then each vertex $v \in S \cup T$ (i.e., a source or a sink)
has exactly one angle larger than $180^\circ$, called a \emph{large
  angle}, in one of its incident faces, and $\deg(v)-1$ angles smaller
than $180^\circ$, called \emph{small angles}, in its other incident
faces.  For a source or sink switch of~$f$, assign either a label~\XL
or a label~\XS to its angle in~$f$, depending on whether this angle is
large or small.  For each face~$f$ of~$D$, the number of \XL-labels
determined by $\Gamma$ equals
$\mycap(f)$~\cite{DBLP:journals/algorithmica/BertolazziBLM94}.
Conversely, suppose given an assignment of \XL- and \XS-labels to the angles
at the source and sink switches of $D$; for each vertex~$v$, $\XL(v)$
(resp.\ $\XS(v)$) denotes the number of \XL- (resp.\ of \XS-) labels
at the angles of~$v$.  For each face~$f$, $\XL(f)$ (resp.\ $\XS(f)$)
denotes the number of \XL- (resp.\ of \XS-) labels at the angles
in~$f$.  Such an assignment corresponds to the labels induced by an
upward planar drawing of~$D$ if and only if the following properties
hold \cite{DBLP:journals/algorithmica/BertolazziBLM94}: (a)~$\XL(v)=0$
for each $v \in I$ and $\XL(v)=1$ for each $v \in S \cup T$;
(b)~$\XL(f)=\mycap(f)$ for each face $f \in F$.  We call such an
assignment an \emph{upward labeling of $D$}, as it uniquely
corresponds to (and hence describes) an upward planar embedding of
$D$; see \cref{fig:updag-1}. We will implicitly assume that a given
upward plane DAG is described by an upward labeling.

Given an upward plane DAG $D$, a \emph{complete saturator} of $D$ is a
set of vertices and edges, not belonging to $D$, used to augment $D$
to a plane $st$-graph $D'$. More precisely, a complete saturator
consists of a source $s$ and a sink $t$, which will belong to the
external face of $D'$, and of a set of edges where each edge $(u,v)$
is called a \emph{saturating edge} and fulfills one of the following
conditions (see, e.g., \cref{fig:updag-2}): (i)~$u,v \notin \{s, t\}$
and $u,v$ are both source switches of the same face $f$ such that $u$
has label~\XS in $f$ and $v$ has label~\XL in $f$; in this case $u$
\emph{saturates} $v$.  (ii)~$u,v \notin \{s, t\}$ and $u,v$ are both
sink switches of the same face $f$ such that $u$ has label~\XL in $f$
and $v$ has label~\XS in $f$; in this case~$v$ \emph{saturates}~$u$.
(iii)~$u=s$ and $v$ is a source switch of the external face with an
\XL angle.  (iv)~$v=t$ and $u$ is a sink switch of the external face
with an \XL angle.

\paragraph{Compaction of turn-regular representations}
We now recall how to compact in linear time a turn-regular rectilinear
representation with the approach of~\cite{DBLP:journals/comgeo/BridgemanBDLTV00}. Consider first any planar rectilinear representation $H$, which is not
necessarily turn-regular. Let $D_x$ be the plane DAG whose vertices
correspond to the maximal vertical chains of $H$ and such that two
vertices of $D_x$ are connected by an edge oriented rightward, if the
corresponding vertical chains are connected by a horizontal segment in
$H$.  Define the upward plane DAG $D_y$ symmetrically, where the
vertices correspond to the maximal horizontal chains of $H$ and where
the edges are oriented upward. Refer to
\cref{fig:turn-regular-compaction}.  $D_x$ and $D_y$ are both upward
plane DAGs (for $D_x$ rotate it by $90^\circ$ to see all edges flowing
in the upward direction).  For a vertex $v$ of $H$, $c_x(v)$ (resp.\
$c_y(v)$) denotes the vertex of $D_x$ (resp.\ of $D_y$) corresponding
to the maximal vertical (resp.\ horizontal) chain of $H$ that
contains~$v$.  For any two vertices $u$ and $v$ of $H$ such that
$c_x(u) \neq c_x(v)$, we write $u \rightsquigarrow_x v$ if there
exists a directed path from $c_x(u)$ to $c_x(v)$ in $D_x$.  We also
write $u \leftrightsquigarrow_x v$ if either $u \rightsquigarrow_x v$
or $v \rightsquigarrow_x u$, while $u \not \leftrightsquigarrow_x v$
means that neither $u \rightsquigarrow_x v$ nor
$v \rightsquigarrow_x u$.  The notations $u \rightsquigarrow_y v$,
$v \rightsquigarrow_y u$, $u \leftrightsquigarrow_y v$, and
$u \not\leftrightsquigarrow_y v$ are used symmetrically referring to
$D_y$ when $c_y(u) \neq c_y(v)$.

Bridgeman et al.~\cite{DBLP:journals/comgeo/BridgemanBDLTV00} showed
that $H$ is turn-regular if and only if, for every two vertices $u$
and $v$ in $H$, we have $u \leftrightsquigarrow_x v$, or
$u \leftrightsquigarrow_y v$, or both.  This is equivalent to saying
that the relative position along the $x$-axis or the relative position
along the $y$-axis (or both) between $u$ and $v$ is fixed over all
drawings of $H$.  Under this condition, the \textsc{OC} problem for
$H$ can be solved by independently solving in $O(n)$ time a pair of 1D
compaction problems for $H$, one in the $x$-direction and the other in
the $y$-direction.  The 1D compaction in the $x$-direction consists
of: $(i)$~augmenting $D_x$ to become a plane $st$-graph by means of a
complete saturator; $(ii)$~computing an optimal topological
numbering~$X$ of~$D_x$ (see~\cite{DBLP:books/ph/BattistaETT99},
p.~89); each vertex $v$ of $H$ receives an x-coordinate $x(v)$ such
that $x(v) = X(c_x(v))$.  We recall that a topological numbering of a
DAG $D$ is an assignment of integer numbers to the vertices of $D$
such that if there is a path from $u$ to $v$ then $u$ is assigned a
number smaller than the number of $v$.  A topological numbering is
optimal if the range of numbers that is used is the minimum possible.
Regarding step $(i)$ of the 1D compaction, note that $D_x$ admits a
unique complete saturator when $H$ is
turn-regular~\cite{DBLP:journals/comgeo/BridgemanBDLTV00}.  This is
due to the absence of kitty corners in each face of~$H$.  The 1D
compaction in the $y$-direction is solved symmetrically, so that each
vertex~$v$ receives a y-coordinate $y(v)=Y(c_y(v))$.
\Cref{fig:turn-regular-compaction} illustrates this process.

\begin{figure}[tb]
  \begin{subfigure}[b]{0.32\linewidth}
    \includegraphics[page=1]{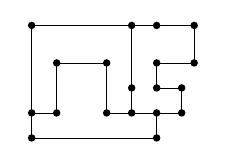}
    \caption{}
    \label{fig:turn-regular-compaction-1}
  \end{subfigure}
  \hfill
  \begin{subfigure}[b]{0.32\linewidth}
    \includegraphics[page=2]{turn-regular-compaction}
    \caption{}
    \label{fig:turn-regular-compaction-2}
  \end{subfigure}
  \hfill
  \begin{subfigure}[b]{0.32\linewidth}
    \includegraphics[page=3]{turn-regular-compaction}
    \caption{}
    \label{fig:turn-regular-compaction-3}
  \end{subfigure}
  \hfill
  \begin{subfigure}[b]{0.32\linewidth}
    \includegraphics[page=4]{turn-regular-compaction}
    \caption{}
    \label{fig:turn-regular-compaction-4}
  \end{subfigure}
  \hfill
  \begin{subfigure}[b]{0.32\linewidth}
    \includegraphics[page=5]{turn-regular-compaction}
    \caption{}
    \label{fig:turn-regular-compaction-5}
  \end{subfigure}
  \hfill
  \begin{subfigure}[b]{0.32\linewidth}
    \includegraphics[page=6]{turn-regular-compaction}
    \caption{}
    \label{fig:turn-regular-compaction-6}
  \end{subfigure}
  \caption{(a) A turn-regular orthogonal representation~$H$.
    (b)--(c)~The maximal horizontal and vertical chains of $H$ are
    highlighted.  (d)~The upward plane DAG $D_x$ with its complete
    saturator (dashed edges) and an optimal topological numbering.
    (e)~The same for $D_y$.  (f)~A minimum-area drawing of $H$ where
    the x- and y-coordinates correspond to the two optimal topological
    numberings.}
  \label{fig:turn-regular-compaction}
\end{figure}

Unfortunately, if $H$ is not turn-regular, the aforementioned approach
fails.  This is because there are in general many potential complete
saturators for augmenting the two upward plane DAGs $D_x$ and $D_y$ to
plane $st$-graphs. Also, even when an $st$-graph for each DAG is
obtained from a complete saturator, computing independently an optimal
topological numbering for each of the two $st$-graphs may lead to
non-planar drawings if no additional relationships are established for
the coordinates of kitty corner pairs, because for a pair $\{u,v\}$ of
kitty corners we have $u \not\leftrightsquigarrow_x v$ and
$u \not\leftrightsquigarrow_y v$.  

\paragraph{Our FPT algorithm}
We now prove that \textsc{OC} is
fixed-parameter-tractable when parameterized by the number of kitty
corners.

\begin{theorem}\label{th:fpt-kitty}
  Given a planar rectilinear representation~$H$ with $n$ vertices and
  $k>0$ kitty corners, we can compute a minimum-area drawing of $H$ in
  $O(2^{8.4k} n \log n)$ time.  In other words, OC is FPT when parameterized by~$k$.
\end{theorem}

\begin{proof}
  For each pair
  $\{u,v\}$ of kitty corners
  in each face,
  we guess the relative positions of $u$ and $v$ in a drawing of $H$,
  i.e., $x(u) \lesseqgtr x(v)$ and $y(u) \lesseqgtr y(v)$.
  This implies $O(k^2)$ guesses.

  More precisely, we generate all maximal plane DAGs (together with an upward
  planar embedding) that can be incrementally obtained from $D_x$ by
  repeatedly applying the following sequence of steps: Guess a pair
  $\{u,v\}$ of kitty corners of $H$ such that $c_x(u)$ and $c_x(v)$
  belong to the same face; for such a pair either add a directed edge
  $(c_x(u),c_x(v))$ (which establishes that $x(u)<x(v)$), or add a
  directed edge $(c_x(v),c_x(u))$ (which establishes that
  $x(u) > x(v)$), or identify $c_x(u)$ and $c_x(v)$ (which establishes
  that $x(u)=x(v)$); this last operation corresponds to adding in $H$
  a vertical segment between $u$ and $v$, thus merging the vertical
  chain of $u$ with the vertical chain of~$v$.  Analogously, we
  generate from~$D_y$ a set of maximal plane DAGs (together with an
  upward planar embedding).  Let $\overline{D_x}$ and $\overline{D_y}$
  be two upward plane DAGs generated from~$D_x$ and~$D_y$,
  respectively, as described above.  We augment
  $\overline{D_x}$ (resp. $\overline{D_y}$) with a complete saturator
  that makes it a plane $st$-graph. Observe that, by construction,
  neither $\overline{D_x}$ nor $\overline{D_y}$ contain two
  non-adjacent vertices in the same face whose corresponding chains of
  $H$ have a pair of kitty corners. Hence their complete saturators
  are uniquely defined. We finally compute a pair of optimal
  topological numberings to determine the $x$- and the $y$-coordinates
  of each vertex of $H$ using the algorithm of Bridgeman et
  al.~\cite{DBLP:journals/comgeo/BridgemanBDLTV00}.
  Note that, for some pairs of $\overline{D_x}$ and $\overline{D_y}$,
  the procedure described above may assign $x$- and $y$-coordinates to
  the vertices of $H$ that do not lead to a planar orthogonal
  drawing. If so, we discard the solution. Conversely, for those
  solutions that correspond to (planar) drawings of $H$, we compute
  the area and, at the end, we keep one of the drawings having minimum
  area.

  \begin{figure}[tbp]
    \begin{subfigure}[b]{0.235\linewidth}
      \centering \includegraphics[page=1]{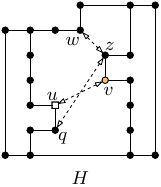}
      \caption{}
      \label{fig:non-turn-regular-compaction-1}
    \end{subfigure}
    \hfill
    \begin{subfigure}[b]{0.35\linewidth}
      \centering \includegraphics[page=2]{non-turn-regular-compaction}
      \caption{}
      \label{fig:non-turn-regular-compaction-2}
    \end{subfigure}
    \hfill
    \begin{subfigure}[b]{0.3\linewidth}
      \centering \includegraphics[page=3]{non-turn-regular-compaction}
      \caption{}
      \label{fig:non-turn-regular-compaction-3}
    \end{subfigure}
	
    \medskip
  
  \begin{subfigure}[b]{1\linewidth}
    \centering \includegraphics[page=4]{non-turn-regular-compaction}
    \caption{}
    \label{fig:non-turn-regular-compaction-4}
  \end{subfigure}
  \caption{(a) An orthogonal representation $H$ with three pairs of
    kitty corners, $\{u,v\}$, $\{w,z\}$, $\{q,z\}$,
    and $k=5$ kitty corners in total.
    (b)~Two distinct (saturated) upward plane DAGs $\overline{D_x}$ and
    $\overline{D_x}'$, and their optimal topological numberings $X$
    and $X'$; in $\overline{D_x}'$ the nodes $c_x(u)$ and $c_x(v)$ are
    identified (filled square).  (c)~Two distinct (saturated) upward
    plane DAGs $\overline{D_y}$ and $\overline{D_y}'$ and their
    optimal topological numberings $Y$ and $Y'$.  (d)~Drawings derived
    from the four different combinations of the topological
    numberings: $\Gamma_1$ and $\Gamma_3$ have sub-optimal areas,
    $\Gamma_2$ has minimum area, and $\Gamma_4$ is non-planar (the
    bold red face is self-crossing).}
  \label{fig:non-turn-regular-compaction}
\end{figure}

\Cref{fig:non-turn-regular-compaction} shows an example of a
non-turn-regular orthogonal representation;
\cref{fig:non-turn-regular-compaction-4} depicts four distinct
drawings resulting from different pairs of upward plane DAGs, each
establishing different $x$- and $y$-relationships between pairs of
kitty corners.  One of the drawings has minimum area; another one is
not planar and therefore discarded.

We now analyze the time complexity of the given algorithm.  Let
$\{f_1, f_2, \dots,$ $ f_h\}$ be the set of faces of $H$, and let
$k_i$ be the number of kitty corners in $f_i$ (i.e., the number of
vertices that are involved in at least one pair of kitty corners).
Since kitty corners are
reflex,
a vertex can be a kitty corner at most
with respect to one face.  Hence, $\sum_{i=1}^h k_i = k$.

Denote by $a_i$ the number of distinct maximal planar augmentations of
$f_i$ with edges that connect pairs of kitty corners. An upper bound
to the value of $a_i$ is the number $c_{k_i}$ of distinct maximal
outerplanar graphs with $k_i$ vertices, which corresponds to the
number of distinct triangulations of a convex polygon with $k_i$
vertices. It is known that $c_{k_i}$ equals the $(k_i-2)$-nd Catalan
number (see, e.g.,~\cite{p-2009}), whose standard estimate is
$c_{k_i-2} \sim \frac{4^{k_i-2}}{(k_i-2)^{3/2}\sqrt{\pi}}$. Therefore,
$a_i \le 4^{k_i}$. Note that, from an algorithmic point of view,
all distinct triangulations of a convex polygon with given set of
vertices can easily be generated with a recursive approach: every time
we guess an edge, it divides an internal face into two faces, and in
each of the two faces we recursively guess the next edge.

Now, for each edge $(u,v)$ of a maximal planar augmentation of $f_i$
such that $\{u,v\}$ is a pair of kitty corners in $H$, we have to
consider three alternative possibilities:
$\overline{D_x}$ has a directed edge $(c_x(u),c_x(v))$, or
$\overline{D_x}$ has a directed edge $(c_x(v),c_x(u))$, or $c_x(u)$
and $c_x(v)$ are identified in $\overline{D_x}$.

Recall that the number of edges of a maximal outerplanar graph on
$k_i$ vertices is $2k_i-3$.
  Hence, for each augmentation of~$H$ that adds $m_i$ edges inside
  each face $f_i$, we have $m_i \leq  2k_i-3$, i.e., $H$ consists of
  $\sum_{i=1}^h m_i \leq \sum_{i=1}^h (2k_i-3) \le 2k$ edges.
Therefore, the number of different possibilities to be considered
for~$\overline{D_{x}}$ is at most $3^{2k}$.  The same holds for
$\overline{D_y}$.
We have to consider at most $a_i \le 4^{k_i}$ augmentations for
face~$f_i$ and at most $\Pi_{i=1}^h 4^{k_i} \le 4^k$ augmentations for
the whole graph~$H$.

The total number of augmentation of~$H$ times the number of
possibilities to extend~$D_x$ to~$\overline{D_{x}}$ times the number
of possibilities to extend~$D_y$ to~$\overline{D_{y}}$ is then
$4^k \cdot 3^{2k} \cdot 3^{2k} = 324^{k} < 2^{8.4k}$.  For each of these
combinations, we augment the two upward plane DAGs~$\overline{D_{x}}$
and~$\overline{D_{y}}$ to plane $st$-graphs using complete saturators
and compute an optimal topological numbering in $O(n)$ time.  Then we
test whether the drawing resulting from the two topological numberings
is planar, which can be done in $O(n \log n)$ time by a sweep-line
algorithm (see, e.g.,~\cite{bo-79,sh-76}).  It follows that the whole
testing algorithm takes
$O(2^{8.4k} n \log n)$ time.
\end{proof}

\section{A Polynomial Kernel for Cycle Graphs}
\label{sec:kernel}

In this section, we give a polynomial kernel for cycle graphs parameterized by the number of kitty corners. To this end, without loss of generality, in this section, we consider \textsc{OC} as a decision problem as follows: Given a planar rectilinear representation $H$ of a connected planar graph $G$ and a rectangle $\cal B$, decide whether there exists an orthogonal drawing of $H$ having bounding box $\cal B$. We prove the following theorem.

\begin{theorem}\label{the:polynomial-kernel}
  Parameterized by the number of kitty corners, \textsc{OC} admits a polynomial kernel on cycle~graphs.
\end{theorem}

Let $[a,b]$ denote the set of natural numbers
$\{a, a+1, \ldots, b-1, b\}$.  We first design a compression with a
linear number of vertices for the \textsc{OC} problem on cycle graphs,
parameterized by the number of kitty corners. We give a compression
algorithm from the \textsc{OC} problem on cycle graphs to the
\textsc{OC} problem on cycle graphs with additional weight constraints
on edges. We call this problem the \emph{weighted orthogonal
  compaction} problem and define it formally as follows. Given a
planar rectilinear representation $H$ of a connected planar graph~$G$
with integer edge weights and a rectangle $\cal B$  on the integer grid and of polynomial size, the weighted orthogonal compaction problem
(\textsc{Weighted OC} for short) asks whether there exists a drawing of $H$ having bounding box $\cal B$ such that the length of each edge in the drawing is at least the weight of the edge. 
At the end of our compression algorithm, by using \cref{pro:compToKern}, we will show how this gives us a polynomial kernel for \textsc{OC}.

Our compression algorithm is presented as a number of reduction
rules. A reduction rule is a polynomial-time procedure that replaces
an instance $(I,k)$ of a parameterized decision problem $\Pi$ by a new instance
$(I',k')$ of another parameterized decision problem $\Pi'$ where
$|I'| \leq |I|$ and $k' \leq k$. The rule is called \emph {safe} if
$(I,k)$ is a \yes\ instance of $\Pi$ if and only if $(I',k')$ is a
\yes\ instance of $\Pi'$.

Let $G$ be a cycle graph and let $H$ be a rectilinear representation
of $G$. As $G$ is a cycle graph, $H$ has only one internal face,
$f_\inter$, which traverses all the vertices and edges of
$G$. Therefore, for ease of writing in this section, whenever we
talk about a face of $H$, we refer to the face $f_\inter$ unless stated
otherwise. We traverse the face $f_\inter$ in the counterclockwise
direction and define a new labeled directed graph $\dirG$ as
follows. We direct every edge $e=(u,v)$ from $u$ to $v$ such that $u$
comes before $v$ while traversing $e$. We label every edge as
\E, \W, \N, or~\South depending on whether the edge is directed in the east,
west, north, or south direction, respectively. Moreover, we label every
vertex as~\F, \R, or~\C if and only if the vertex is a flat, reflex, or
convex vertex, respectively. See \cref{fig:dirG}. For a vertex~$v$,
let $\lab(v)$ be the label of~$v$ in $\dirG$. We say that $v$ is a
$\lab(v)$-vertex. Similarly, for an edge~$e$, let $\lab(e)$ be the
label of~$e$ in $\dirG$. We say that $e$ is a $\lab(e)$-edge. For any
two vertices $u$ and $v$, let $P_{u,v}$ be the directed path from $u$
to $v$ in $\dirG$.
In the rest of this section, we talk about the graph $\dirG$, unless stated
otherwise. Moreover, when we talk about a drawing $\Gamma$ of $\dirG$, we
refer to a planar (orthogonal) drawing of $H$. Furthermore, when we talk
about a labeling of a path of $\dirG$, we talk about vertex labeling,
unless stated otherwise.

As every vertex $v$ in $G$ has degree $2$, $v$ has exactly one
incoming edge and one outgoing edge in $\dirG$. Let $\pred(v)$ and
$\need(v)$ be the incoming edge to $v$ and the outgoing edge from $v$,
respectively. We call the ordered pair
$(\lab(\pred(v)),\lab(\need(v)))$ the \emph{edge label pair
  associated} to $v$. Observe that $v$ is an $\F$-vertex if and only
if its associated edge label pair belongs to the set
$\{(\E,\E), (\W,\W), (\N,\N), (\South,\South)\}$.  Similarly, $v$ is
an \R-vertex (resp., a \C-vertex) if and only if its associated edge
label pair belongs to the set
$\{(\W,\N), (\N,\E), (\E,\South), (\South,\W)\}$ (resp.,
$\{(\W,\South), (\South,\E), (\E,\N), (\N,\W)\}$).

We denote by $(v,\E)$, $(v,\W)$, $(v,\N)$ and $(v,\South)$ the ray
with the endpoint $v$ and going in the east, west, north and south
direction, respectively. Let each of $\ell_1$ and $\ell_2$ be a ray or
a line segment. We denote by $\intP(\ell_1,\ell_2)$ the intersection
point of $\ell_1$ and $\ell_2$ (if it exists). Given an edge $e$, we
denote by $\len(e)$ the weight of the edge $e$. We initialize all edge weights to be $1$.

\begin{figure}[tb]
  \centering
  \begin{subfigure}{0.31\textwidth}
    \includegraphics[page=1]{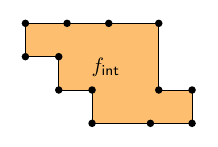}
    \caption{}
    \label{fig:dirG-1}
  \end{subfigure}
  \hfill
  \begin{subfigure}{0.31\textwidth}
    \includegraphics[page=2]{kittyCornersKernel}
    \caption{}
    \label{fig:dirG-2}
  \end{subfigure}
  \hfill
  \begin{subfigure}{0.31\textwidth}
    \includegraphics[page=3]{kittyCornersKernel}
    \caption{}
    \label{fig:dirG-3}
  \end{subfigure}
  \caption{An illustration of (a)~a cycle graph $G$ and the
    corresponding (b)~labeled directed edges and (c)~labeled vertices
    of $\dirG$.}
  \label{fig:dirG}
\end{figure}

Let $\langle c_1, c_2, \ldots, c_k, c_{k+1} = c_1 \rangle$ be the
cyclic counterclockwise order of kitty corners of $H$ in $\dirG$. We look at the path
$P_{c_i, c_{i+1}}$, for every $i \in [1,k]$, and bound the number of
internal vertices of this path as a function of $k$. As $\dirG$ is the
union of all such paths, this, in turn, will bound the size of the
reduced instance as a function of $k$.
We will give a series of reduction rules to reduce the number of
vertices of such paths. We will always apply the rules in the order
they are given. This, in turn, implies that when we apply some
Reduction Rule~$i$ on an instance, no other Reduction Rule~$j$ with
$j < i$ can be further applied to the instance.

We start with a simple reduction rule that reduces a path $P$ on
$\F$-vertices to a weighted edge. Formally, we have the following
rule.

\begin{reduction}\label{rer:flatPath}
  Suppose that there exists a path $P = (u_1, u_2, \ldots, u_t)$ in
  $\dirG$ such that $t \geq 3$ and $\lab(u_i) = \F$ for every
  $i \in [2,k-1]$. Then, delete the path $(u_2, \ldots, u_{t-1})$ and
  connect $u_1$ to $u_t$ by a directed edge from $u_1$ to $u_t$ whose
  weight is the sum of the weights of the edges of $P$.
\end{reduction}

This reduction rule is safe because a path on $\F$-vertices is always
drawn as a straight line path between its end points in any drawing of
$\dirG$ such that the length of the path is at least its weight. So we
can replace the path with a straight-line edge between its end points
whose weight is the sum of the weights of the edges of the path, and
vice-versa.

As we apply the reduction rules in order, due to \cref{rer:flatPath},
for the rest of the section, we assume that $\dirG$ does not have any
$\F$-vertex. So, all the internal vertices of $P_{c_i,c_{i+1}}$, for
any $i \in [1,k]$, are \R- or \C-vertices and none of them are kitty
corners of $H$. Let
$P_{c_i,c_{i+1}} = (v_1=c_i, v_2, \ldots, v_t=c_{i+1})$. If
$t \leq 16$, we do not reduce the path $P_{c_i,c_{i+1}}$. Otherwise,
we reduce the truncated path
$P^{\trun}_{c_i,c_{i+1}} = (v_6, v_7, \ldots, v_{t-5})$, leaving five
buffer vertices on both the sides. These buffer vertices will be
helpful later in our reduction rules. Note that the number of vertices
of $P^{\trun}_{c_i,c_{i+1}}$ is at least $7$ for any $i \in [1,k]$. 
We now give the following observation about $P^{\trun}_{c_i,c_{i+1}}$,
which will be useful in designing our reduction rules.

\begin{observation}\label{obs:patternMatching}
  Let $i \in [1,k]$. The path $P^{\trun}_{c_i,c_{i+1}}$ has either
  only \R-vertices or a subpath whose labeling belongs to the set
  $\{\RCR, \RCCC, \RCCR, \RRRC,$ $\CRC, \CRRR,$ $\CRRC, \CCCR\}$. See
  \cref{fig:patternMatching}.
\end{observation}
\begin{figure}[tb]
  \centering \includegraphics[page=27]{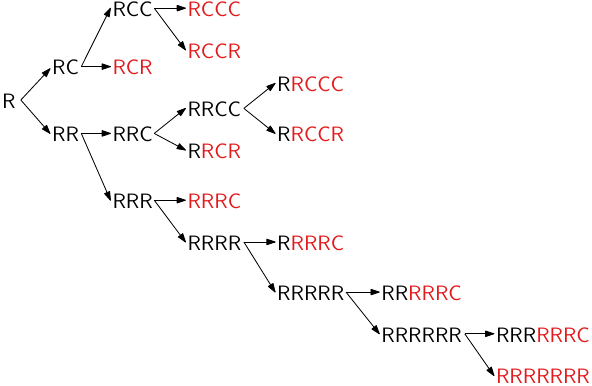}
  \caption{An illustration showing that a path on $7$ or more vertices
    starting with an \R-vertex has either only \R-vertices or a
    subpath (shown in red) whose labeling belongs to the set
    $\{\RCR, \RCCC, \RCCR, \RRRC\}$. The case of a path starting with a \C-vertex is symmetric and thus not shown in the figure.} 
  \label{fig:patternMatching}
\end{figure}

Due to the above observation, we only give reduction rules for the
(sub)paths whose labeling belongs to the set
$\{\RCR, \RCCC, \RCCR, \RRRC,\CRC,$ $\CRRR,$ $\CRRC, \CCCR\}$. Before
we start with our reduction rules, given an \R- or a \C-vertex $v$ and
a drawing $\Gamma$ of $\dirG$, we define two vertices $\nx(v,\Gamma)$ and
$\ny(v,\Gamma)$, called the \emph{nearest y-vertex} and the \emph{nearest
  x-vertex} of $v$ in $\Gamma$, respectively. Note that $\nx(v,\Gamma)$ can be
the same as $\ny(v,\Gamma)$. Intuitively, $\nx(v,\Gamma)$ (resp., $\ny(v,\Gamma)$) is
the vertex of $\dirG$ in the minimum size rectangle bounding
$v, \pred(v)$ and $\need(v)$, which is ``nearest'' to $v$ in the
x-direction (resp., y-direction) and
$\lab(\nx(v,\Gamma)) \in \{\R,\C\} \setminus \{\lab(v)\}$ (resp.,
$\lab(\ny(v,\Gamma)) \in \{\R,\C\} \setminus \{\lab(v)\}$) (see \cref{fig:nearXnY}).

Although the total number of different edge label pair associated with
an \R- or a \C-vertex is $8$, for defining the above two vertices we
need to look at only $4$ pairs of edge label pairs corresponding to
either an \R- or a \C-vertex as the other $4$ are symmetric in the
following sense. If we describe the neighbors of a vertex $v$ as to
the right of $v$ and above $v$, then $v$ is either an \R-vertex with
associated edge label pair $(\W,\N)$ or a \C-vertex with associated
edge label pair $(\South,\E)$. Similarly, $(\N,\E), (\E,\South)$ and
$(\South,\W)$ are symmetric to $(\W,\South), (\N,\W)$ and $(\E,\N)$,
respectively. So, given a \C- or \R-vertex $v$, let
$\lef(v), \rig(v), \abo(v), \bel(v)$ be the vertex which is to the
left, to the right, above and below $v$, if it exists. Observe that
given a \C- or \R-vertex $v$ exactly two out of
$\lef(v), \rig(v), \abo(v), \bel(v)$ exist. In what follows, we give
observations about the existence of $\nx(v,\Gamma)$ and $\ny(v,\Gamma)$ and some
of their properties (see \cref{fig:nearXnY}).

\begin{figure}[tb]
  \centering
  \begin{subfigure}{0.24\textwidth}
    \includegraphics[page=4]{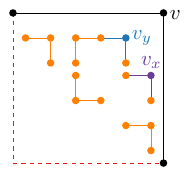}
    \caption{}
    \label{fig:nearXnY-1}
  \end{subfigure}
  \hfill
  \begin{subfigure}{0.24\textwidth}
    \includegraphics[page=5]{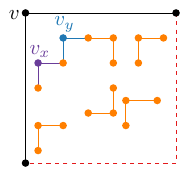}
    \caption{}
    \label{fig:nearXnY-2}
  \end{subfigure}
  \hfill
  \begin{subfigure}{0.24\textwidth}
    \includegraphics[page=6]{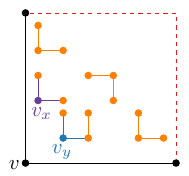}
    \caption{}
    \label{fig:nearXnY-3}
  \end{subfigure}
  \hfill
  \begin{subfigure}{0.24\textwidth}
    \includegraphics[page=7]{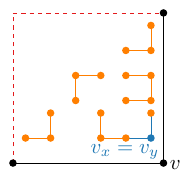}
    \caption{}
    \label{fig:nearXnY-4}
  \end{subfigure}
  \caption{An illustration for the definitions of $\nx(v, \Gamma)$ (shown by
    $v_x$) and $\ny(v, \Gamma)$ (shown by $v_y$). The corresponding
    definitions are given in
    Observations~\ref{obs:LB},~\ref{obs:RB},~\ref{obs:RA}
    and~\ref{obs:LA} are illustrated in~(a),~(b),~(c) and~(d),
    respectively.}
  \label{fig:nearXnY}
\end{figure}

\begin{observation}\label{obs:LB}
  Let $v$ be an \R- or a \C-vertex of $\dirG$ such that $\lef(v)$ and
  $\bel(v)$ exist. Let $\Gamma$ be a drawing of $\dirG$ such that there
  exists a vertex $u$ for which $(i)$
  $x(u) \in [x(\lef(v)),x(\bel(v))]$, $(ii)$
  $y(u) \in [y(\bel(v)), y(\lef(v))]$, and $(iii)$
  $u, \pred(u), \need(u) \notin \{\lef(v), \bel(v)\}$. Let $S$ be the
  set of all such vertices. Let $\maY = \max\{y(v) \mid v \in S\}$ and
  $\maX = \max\{x(v) \mid v \in S\}$.
	
  Then, $\nx(v,\Gamma)$ is the vertex in $S$ such that $x(\nx(v,\Gamma)) = \maX$
  and $y(\nx(v,\Gamma)) = \max\{y(v) \mid v \in S \wedge x(v) =
  \maX\}$. Similarly, $\ny(v,\Gamma)$ is the vertex in $S$ such that
  $y(\ny(v,\Gamma)) = \maY$ and
  $x(\ny(v,\Gamma)) = \max\{x(v) \mid v \in S \wedge y(v) = \maY\}$. It
  follows from the definition that the neighbors of $\nx(v,\Gamma)$ (resp.,
  $\ny(v,\Gamma)$) are to the left of and below $\nx(v,\Gamma)$ (resp.,
  $\ny(v,\Gamma)$). Moreover, from the planarity of $\Gamma$, it follows that
  $\lab(\nx(v,\Gamma)) \in \{\R,\C\} \setminus \{\lab(v)\}$ and
  $\lab(\ny(v,\Gamma)) \in \{\R,\C\} \setminus \{\lab(v)\}$. See
  \cref{fig:nearXnY-1}.
\end{observation}

As the observation for the remaining cases, i.e. when either $\rig(v)$ and $\bel(v)$ (\cref{fig:nearXnY-2}), or $\rig(v)$ and $\abo(v)$ (\cref{fig:nearXnY-3}), or $\lef(v)$ and $\abo(v)$ (\cref{fig:nearXnY-4}) exist, are similar to \cref{obs:LB}, they are given as Observations~\ref{obs:RB}--\ref{obs:LA} in the \hyperref[app:kernel]{Appendix}. For ease of writing, if drawing is clear from the context, we may refer $\nx(v,\Gamma)$ and $\ny(v,\Gamma)$ as $\nx(v)$ and $\ny(v)$, respectively.

Let $v$ be an \R- or a \C-vertex of $\dirG$. Based on these
observations, we now give a lemma about the existence of kitty corner
pairs $(c,c')$, where $c \in \{\pred(v), \need(v)\}$ and
$c' \in \{\nx(v,\Gamma), \ny(v,\Gamma)\}$.

\begin{lemma}\label{lem:kittyCorner}
  Let $v$ be an \R- or a \C-vertex of $\dirG$. Let $\Gamma$ be a drawing of
  $\dirG$ such that $\nx(v,\Gamma)$ and $\ny(v,\Gamma)$ exist. If
  $\lab(\pred(v)) =$ $ \lab(\nx(v,\Gamma))$ (resp., $\lab(\need(v))$ $=
  \lab(\nx(v,\Gamma)))$, then $(\pred(v),$ $\nx(v,\Gamma))$ and $(\pred(v),$
  $\ny(v,\Gamma))$ (resp., $(\need(v), \nx(v,\Gamma))$ and
  $(\need(v), \ny(v,\Gamma)))$ are kitty corners of $H$.
\end{lemma}

\begin{figure}[tb]
  \centering
  \begin{subfigure}{0.49\textwidth}
    \centering \includegraphics[page=8]{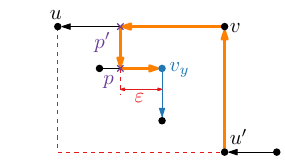}
    \caption{}
    \label{fig:kittyCorner-1}
  \end{subfigure}
  \hfill
  \begin{subfigure}{0.49\textwidth}
    \centering \includegraphics[page=9]{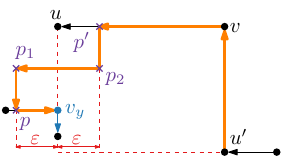}
    \caption{}
    \label{fig:kittyCorner-2}
  \end{subfigure}
  \caption{Illustration for the two cases considered in the proof of
    \cref{lem:kittyCorner} based on the relative positions of
    $\need(v)$ (shown by $u$) and $\ny(v,\Gamma)$ (shown by $v_y$). (a) The
    case where $x(\ny(v,\Gamma)) > x(\need(v))$. (b) The case where
    $x(\ny(v,\Gamma)) = x(\need(v))$. The paths from $\pred(v)$ (shown by
    $u'$) to $\ny(v,\Gamma)$, to show that they are a pair of kitty
    corners, are drawn in orange.}
  \label{fig:kittyCorner}
\end{figure}

\begin{proof}
  We prove the lemma for $\pred(v)$ and $\ny(v,\Gamma)$ for a \C-vertex
  $v$. The proofs for the other cases are similar. Without loss of
  generality, we assume that the edge label pair associated with $v$
  is $(\N,\W)$. This, in turn, implies that
  $\lab(\pred(v)) = \lab(\ny(v,\Gamma)) = \R$ and the edge label pair
  associated with $\ny(v,\Gamma)$ is $(\E,\South)$. Moreover, there also
  exists a path from $\ny(v,\Gamma)$ to $v$. We consider two cases based on
  whether $x(\ny(v,\Gamma)) = x(\need(v))$ or $x(\ny(v,\Gamma)) > x(\need(v))$.
	
  Let $\varepsilon \in \mathbb{R}$ such that $0 <\varepsilon < 1$. Let
  $p$ be a point that is immediately to the left of $\ny(v,\Gamma)$ on the
  line segment corresponding to the edge $(\pred(\ny(v,\Gamma)),\ny(v,\Gamma))$ in $\Gamma$,
  such that $x(p) = x(\ny(v,\Gamma)) - \varepsilon$. We first consider the
  case where $x(\ny(v,\Gamma)) > x(\need(v))$ (see
  \cref{fig:kittyCorner-1}). Let $p'$ be the point on the line segment
  corresponding to the edge $(v,\need(v))$ in $\Gamma$ such that
  $x(p') = x(p)$. From the definitions of $\ny(v,\Gamma)$, $p$ and $p'$, we
  have that the line segment $(p,p')$ does not intersect any line segment in $\Gamma$. We
  now focus on the directed cycle graph $\cal C$ formed by the path
  from $\ny(v,\Gamma)$ to $v$ in $\dirG$ followed by the path
  $(v,p',p,\ny(v,\Gamma))$ where $\lab(p) = \lab(p') = \C$,
  $\lab((v,p')) = W$, $\lab((p',p)) = \South$,
  $\lab((p,\ny(v,\Gamma)) = \E$ and the labels of all the other vertices
  and edges remain as in $\dirG$. Observe that by the construction,
  there exists a planar drawing of $\cal C$. As
  $\rot(\pred(v),\ny(v,\Gamma)) = (-1) + 1 + 1 + 1 = 2$,
  $\rot(\ny(v,\Gamma), \pred(v)) = 2$. As the path from $\ny(v,\Gamma)$ to
  $\pred(v)$ in $\cal C$ is the same as that in $\dirG$,
  $(\ny(v,\Gamma), \pred(v))$ is a pair of kitty corners of $H$.
	
  We now consider the case where $x(\ny(v,\Gamma)) = x(\need(v))$ (see
  \cref{fig:kittyCorner-2}). Let $p_1$ be a point such that
  $x(p_1) = x(p)$ and $y(p_1) = y(p) + \varepsilon$. Let $p_2$ be a
  point such that $x(p_2) = x(\ny(v,\Gamma)) + \varepsilon$ and
  $y(p_2) = y(p_1)$. Let $p'$ be the point on the line segment corresponding
  to the edge $(v,\need(v))$ in $\Gamma$ such that $x(p') = x(p_2)$. From
  the definitions of $\ny(v,\Gamma)$, $p$, $p_1$, $p_2$ and $p'$, we have
  that the line segments $(p,p_1)$, $(p_1,p_2)$ and $(p_2,p')$ do not
  intersect any line segment in $\Gamma$. We now focus on the directed cycle graph
  $\cal C$ formed by the path from $\ny(v,\Gamma)$ to $v$ in $\dirG$
  followed by the path $(v,p',p_2,p_1,p,\ny(v,\Gamma))$ where
  $\lab(p) = \lab(p_1) = \lab(p') = \C$, $\lab(p_2) = \R$,
  $\lab((v,p')) = W$, $\lab((p',p_2)) = \South$,
  $\lab((p_2,p_1)) = W$, $\lab((p_1,p)) = \South$,
  $\lab((p,\ny(v,\Gamma)) = \E$, and the labels of all the other vertices
  and edges remain as in $\dirG$. Observe that by the construction,
  there exists a planar drawing of $\cal C$. As
  $\rot(\pred(v),\ny(v,\Gamma)) = (-1) + 1 + 1 + (-1) + 1 + 1 = 2$,
  $\rot(\ny(v,\Gamma), \pred(v)) = 2$. As the path from $\ny(v,\Gamma)$ to
  $\pred(v)$ in $\cal C$ is the same as that in $\dirG$,
  $(\ny(v,\Gamma), \pred(v))$ is a pair of kitty corners of $H$.
\end{proof}

Without loss of generality, in the rest of the section, we assume that
the first edge of the path we want to reduce is a
$\mathsf{W}$-edge. We remark that a strategy based on ``rectangle
cutting'' that somewhat resembles our reduction rules has been
employed by Tamassia~\cite{t-eggmn-87} for a different purpose. We now
give our reduction rules for a path labeled $\RCR$ or $\CRC$.  We give the reduction rule for $\RCR$. As the reduction rule for $\CRC$ is similar, it is given as \cref{rer:CRC} in the \hyperref[app:kernel]{Appendix}. 

\begin{reduction}\label{rer:RCR}
  Suppose that there exists a path $P= (u_1, u_2, u_3,u_4,u_5)$ in
  $\dirG$ such that $\lab(u_2) = \R$, $\lab(u_3) = \C$ and
  $\lab(u_4) = \R$. Then, delete the vertex $u_3$ and the edges
  incident to it from $\dirG$ and add a new path $(u_2, u'_3, u_4)$ to
  $\dirG$. Moreover, assign $\lab(u_2) = \F$, $\lab(u'_3) = \R$,
  $\lab(u_4) = \F$, $\len((u_2,u'_3)) = \len((u_3,u_4))$ and
  $\len((u'_3,u_4)) = \len((u_2,u_3))$. The labels of all the
  remaining vertices and the weights of all the remaining edges stay
  the same. Let $\dirG_\redu$ be the reduced graph. See
  \cref{fig:RCR}.
\end{reduction}

\begin{figure}[tb]
  \centering
  \begin{subfigure}{0.32\textwidth}
    \includegraphics[page=10]{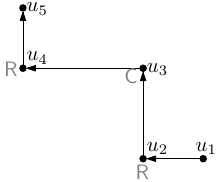}
    \caption{}
    \label{fig:RCR-1}
  \end{subfigure}
  \hfill
  \begin{subfigure}{0.32\textwidth}
    \includegraphics[page=11]{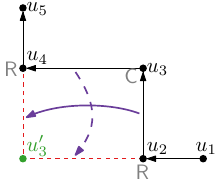}
    \caption{}
    \label{fig:RCR-2}
  \end{subfigure}
  \hfill
  \begin{subfigure}{0.32\textwidth}
    \includegraphics[page=12]{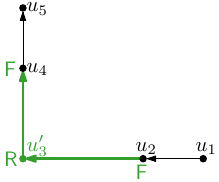}
    \caption{}
    \label{fig:RCR-3}
  \end{subfigure}
  \caption{An illustration for \cref{rer:RCR}. The original path $P$
    is shown in (a) and the reduced path is shown in~(c). (b)~shows
    the vertex $u'_3$, and the projection from the edges of $\dirG$
    that will be deleted when applying the reduction rule to the edges
    of $\dirG_\redu$. The projection is shown by curved purple solid
    and dotted edges from an edge to be deleted to its projected
    edge.}
  \label{fig:RCR}
\end{figure}

\begin{lemma}
  \label{lem:RCR}
  \cref{rer:RCR} is safe.
\end{lemma}

\begin{figure}
  \centering
  \begin{subfigure}{0.48\textwidth}
    \centering \includegraphics[page=13]{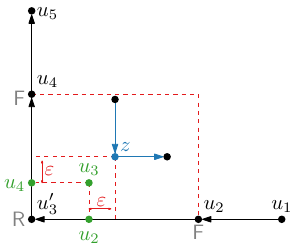}
    \caption{}
    \label{fig:proofRCR-1}
  \end{subfigure}
  \hfill
  \begin{subfigure}{0.48\textwidth}
    \centering \includegraphics[page=14]{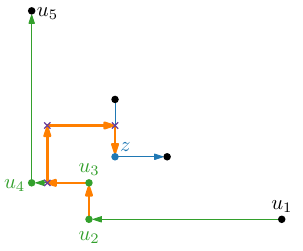}
    \caption{}
    \label{fig:proofRCR-2}
  \end{subfigure}
  \caption{An illustration for the case considered in the reverse
    direction proof of \cref{lem:RCR}. A path from $u_2$ to
    $\ny(u'_3,\Gamma)$ (shown by $z$), to show that they are a pair of
    kitty corners, are drawn in orange.}
  \label{fig:proofRCR}
\end{figure}

\begin{proof}
  To prove that \cref{rer:RCR} is safe, we need to prove that there
  exists a planar drawing $\Gamma$ of $\dirG$ if and only if there exists a
  planar drawing $\Gamma_\redu$ of $\dirG_\redu$ such that both $\Gamma$ and
  $\Gamma_\redu$ have the same bounding box. Recall that, without loss of
  generality, we assume that $\lab((u_1,u_2)) = W$.

  \smallskip

  $(\Rightarrow)$ Let $\Gamma$ be a planar drawing of $\dirG$ having
  bounding box $\cal B$. We get a drawing $\Gamma_\redu$ of $\dirG_\redu$
  having the same bounding box $\cal B$ as follows. For every vertex
  $v \in V(\dirG_\redu) \setminus \{u'_3\}$, $x(v)$ and $y(v)$ in
  $\Gamma_\redu$ are the same as those in $\Gamma$. For $u'_3$, we assign $x(u'_3) =
  x(u_4)$ and $y(u'_3) = y(u_2)$. Observe that $\cal B$ is a bounding
  box of $\Gamma_\redu$. Moreover, if there does not exist any vertex $w
  \in V(\dirG)$ such that $x(w) \in [x(u_4),x(u_2)]$ and $y(w) \in
  [y(u_2),y(u_4)]$ in $\Gamma$, then $\Gamma_\redu$ is a planar drawing of
  $\dirG_\redu$. Assume for contradiction there exists such a
  $w$. Observe that $\pred(w), \need(w) \notin \{u_2,u_4\}$. Then, by
  \cref{obs:LB} and \cref{lem:kittyCorner}, $\ny(u_3)$ exists and
  $(\ny(u_3),u_2)$ is a kitty corner pair of $H$, a contradiction to
  the fact that the path that we are reducing does not contain any
  kitty corner (as it is a path between two consecutive kitty corner vertices).
  
  \smallskip

  $(\Leftarrow)$ Let $\Gamma_\redu$ be a planar drawing of $\dirG_\redu$
  having bounding box $\cal B$. We get a drawing $\Gamma$ of $\dirG$ having
  the same bounding box $\cal B$ as follows. For every vertex
  $v \in V(\dirG) \setminus \{u_3\}$, $x(v)$ and $y(v)$ in $\Gamma$ are the
  same as those in $\Gamma_\redu$. For $u_3$, we assign $x(u_3) = x(u_2)$
  and $y(u_3) = y(u_4)$. Observe that $\cal B$ is a bounding box of
  $\Gamma$. Moreover, if there does not exist any vertex
  $w \in V(\dirG_\redu)$ such that $x(w) \in [x(u_4),x(u_2)]$ and
  $y(w) \in [y(u_2),y(u_4)]$ in $\Gamma_\redu$, then $\Gamma$ is a planar
  drawing of $\dirG$. Assume for contradiction there exists such a
  $w$. Observe that $\pred(w), \need(w) \notin \{u_2,u_4\}$. Then, by
  \cref{obs:RA}, $\ny(u'_3)$ exists and $\lab(\ny(u'_3)) = \C$. Let
  $\Gamma'$ be a drawing of $\dirG$ obtained from $\Gamma_\redu$ as follows. Let
  $\varepsilon \in \mathbb{R}$ such that $0 <\varepsilon < 1$. For
  every vertex $v \in V(\dirG) \setminus \{u_2,u_3,u_4\}$, $x(v)$ and
  $y(v)$ in $\Gamma'$ are the same as those in $\Gamma_\redu$. For $u_2$, we assign
  $x(u_2) = x(\ny(u'_3)) - \varepsilon$ and $y(u_2) = y(u'_3)$. For
  $u_4$, we assign $x(u_4) = x(u'_3)$ and
  $y(u_4) = y(\ny(u'_3)) - \varepsilon$. For $u_3$, we assign
  $x(u_3) = x(u_2)$ and $y(u_3) = y(u_4)$. See
  \cref{fig:proofRCR-1}. 
  Observe that $\Gamma'$ is a planar drawing of
  $\dirG$.
  Similarly to the proof of
  \cref{lem:kittyCorner}, we can prove that $u_3$ and $\ny(u'_3)$
  (which is also a vertex of $\dirG$) make a pair of kitty corner on
  the outer face of $H$ (see \cref{fig:proofRCR-2}), a contradiction
  to the fact that the path that we are reducing does not contain any
  kitty corner.
\end{proof}

We will next give our reduction rules for paths labeled $\RCCR$,
$\CRRC$, $\RCCC$, $\CCCR$, $\CRRR$ and $\RRRC$.  Towards that, we give
some properties of the drawing of these paths in any drawing of
$\dirG$. We first consider a path labeled $\RCCR$.  

\begin{figure}[tb]
  \centering
  \begin{subfigure}{0.48\textwidth}
    \centering \includegraphics[page=15]{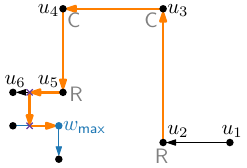}
    \caption{}
    \label{fig:length-1}
  \end{subfigure}
  \hfill
  \begin{subfigure}{0.48\textwidth}
    \centering \includegraphics[page=16]{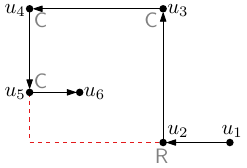}
    \caption{}
    \label{fig:length-2}
  \end{subfigure}
  \caption{An illustration for cases considered in the proofs of
    (a)~\cref{lem:lenRCCR} and (b)~\cref{lem:lenRCCC}. A path from
    $u_2$ to $w_\ma$, to show that they are a pair of kitty corners,
    is drawn in orange.}
  \label{fig:length}
\end{figure}

\begin{lemma}\label{lem:lenRCCR}
  Suppose that there exists a path $P = (u_1,u_2,u_3,u_4,u_5,u_6)$ in
  $\dirG$ such that $\lab(u_2) = \lab(u_5) = \R$ and
  $\lab(u_3) = \lab(u_4) = \C$. For any drawing $\Gamma$ of $\dirG$, there
  exists another drawing $\Gamma'$ of $\dirG$ whose bounding box is the
  same as that of $\Gamma$ and such that $y(u_2) = y(u_5)$ in $\Gamma'$.
\end{lemma}

\begin{proof}
  Let $\Gamma$ be a drawing of $\dirG$. Assume that either
  $y(u_2) < y(u_5)$ or $y(u_5) < y(u_2)$ in $\Gamma$, else we are done (set
  $\Gamma' = \Gamma$). Specifically, we assume that $y(u_2) < y(u_5)$, since the
  case when $y(u_5) < y(u_2)$ is symmetric.
	
  Observe that we can assume that there exists a vertex $w$ such that
  $x(w) \in [x(u_6), x(u_5)]$ and $y(w) = y(u_5) - 1$, since otherwise
  we can shift $u_5$ and $u_6$ downwards. As $y(u_2) < y(u_5)$,
  $y(u_2) \leq y(w)$. Let $S$ be the set of all such vertices $w$. Let
  $w_\ma$ be the vertex in $S$ such that
  $x(w_\ma) = \max\{x(w) \mid w \in S\}$. Observe that by
  \cref{obs:LB} and \cref{lem:kittyCorner}, if there exists a vertex
  $z$ such that $x(z) \in [x(u_5), x(u_2)]$,
  $y(z) \in [y(u_2), y(u_3)]$ and
  $\pred(z), \need(z) \notin \{u_2,u_4\}$, then $u_2$ is a kitty
  corner of $H$. Therefore, there cannot be any such vertex $z$, as
  the path that we are reducing does not contain any kitty corner. Due
  to this, the neighbors of $w_\ma$ are to the left of and below
  $w_\ma$. Moreover, from the planarity of $\Gamma$, we get that
  $\lab(w_\ma) = \R$. Similarly to the proof of
  \cref{lem:kittyCorner}, we can prove that $(w_\ma, u_2)$ is a pair
  of kitty corners of $H$, again a contradiction to the fact that the
  path that we are reducing does not contain any kitty corner. See
  \cref{fig:length-1}. This, in turn, implies that $y(u_5) = y(u_2)$
  in $\Gamma$ or we can get another drawing $\Gamma'$ of $\dirG$ where we can
  shift $u_5$ and $u_6$ downwards such that $y(u_5) = y(u_2)$.
\end{proof}

A lemma similar to \cref{lem:lenRCCR} for a path labeled $\CRRC$ is given as \cref{lem:lenCRRC} in the \hyperref[app:kernel]{Appendix}.

We now give the following lemma for a path labeled $\RCCC$ or $\CCCR$.

\begin{lemma}\label{lem:lenRCCC}
  Suppose that there exists a path $P = (u_1,u_2,u_3,u_4,u_5,u_6)$ in
  $\dirG$ such that $\lab(u_2) = \R$ and
  $\lab(u_3) = \lab(u_4) = \lab(u_5) = \C$ (resp.,
  $\lab(u_2) = \lab(u_3) = \lab(u_4) =\C$ and $\lab(u_5) = \R$). For
  any drawing $\Gamma$ of $\dirG$, $y(u_5) < y(u_2)$ in $\Gamma$.
\end{lemma}

\begin{proof}
  We will prove the lemma for the $\RCCC$ case. The proof for the
  other case is symmetric. Let $\Gamma$ be a drawing of $\dirG$. Assume for
  contradiction that $y(u_2) \leq y(u_5)$. Then, $x(u_6) = x(u_5) + 1$
  and $y(u_6) \in [y(u_2), y(u_4)]$. Moreover,
  $\pred(u_6), \need(u_6) \notin \{u_2,u_4\}$. Then, by \cref{obs:LB}
  and \cref{lem:kittyCorner}, $\ny(u_3)$ exists and $(\ny(u_3),$
  $u_2)$ is a kitty corner pair of $H$, a contradiction to the fact
  that the path that we are reducing does not contain any kitty
  corner. See \cref{fig:length-2}. This implies that $y(u_5) < y(u_2)$
  in $\Gamma$.
\end{proof}

A lemma similar to \cref{lem:lenRCCC} for paths labeled $\CRRR$ or $\RRRC$ is given as \cref{lem:lenCRRR} in the \hyperref[app:kernel]{Appendix}.

We now give the reduction rules for paths labeled $\RCCC$, $\CCCR$,
$\CRRR$ and $\RRRC$, followed by those for paths labeled $\RCCR$ and
$\CRRC$.  We start by giving the reduction rule for $\RCCC$.  Recall
that, due to previous reduction rules,
there is no $\F$-vertex or a path labeled $\RCR$ or
$\CRC$ in $\dirG$.

\begin{reduction}\label{rer:RCCC}
  Suppose that there exists a path $P= (u_1, u_2, u_3,u_4,u_5,u_6)$ in
  $\dirG$ such that $\lab(u_2) = \R$ and
  $\lab(u_3) = \lab(u_4) = \lab(u_5) = \C$. Then, delete the vertices
  $u_3$ and $u_4$ and the edges incident to them from $\dirG$ and add
  a new path $(u_2, u'_3, u_5)$ to $\dirG$. Moreover, assign
  $\lab(u_2) = \F$, $\lab(u'_3) = \C$,
  $\len((u_2,u'_3)) = \len((u_3,u_4))$,
  $\len((\pred(u_1),u_1)) = \max\{\len((u_2,u_3)), \len((\pred(u_1),$ $u_1))\}$, and
  $\len((u'_3,u_5)) = \max\{$ $\len((u_4,$
  $u_5)) - \len((u_2,u_3)), 1\}$. The labels of all the remaining
  vertices and the weights of all the remaining edges stay the
  same. Let $\dirG_\redu$ be the reduced graph. See \cref{fig:RCCC}.
\end{reduction}

\begin{figure}[tb]
  \centering
  \begin{subfigure}{0.32\textwidth}
    \includegraphics[page=17]{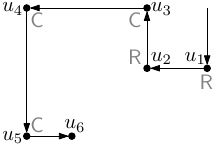}
    \caption{}
    \label{fig:RCCC-1}
  \end{subfigure}
  \hfill
  \begin{subfigure}{0.32\textwidth}
    \includegraphics[page=18]{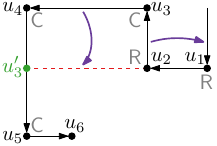}
    \caption{}
    \label{fig:RCCC-2}
  \end{subfigure}
  \hfill
  \begin{subfigure}{0.32\textwidth}
    \includegraphics[page=19]{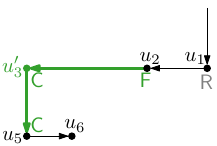}
    \caption{}
    \label{fig:RCCC-3}
  \end{subfigure}
  \caption{An illustration for \cref{rer:RCCC}. The original path $P$
    is shown in~(a) and the reduced path is shown in~(c). (b)~shows
    the vertex~$u'_3$, and the projection from the edges of $\dirG$
    that will be deleted when applying the reduction rule to the edges
    of $\dirG_\redu$, represented by curved purple edges from an edge
    to be deleted to its projected edge.}
  \label{fig:RCCC}
\end{figure}

\begin{lemma}
  \label{lem:RCCC}
  \cref{rer:RCCC} is safe.
\end{lemma}

\begin{proof}
  To prove that \cref{rer:RCCC} is safe, we need to prove that there
  exists a planar drawing $\Gamma$ of $\dirG$ if and only if there exists a
  planar drawing $\Gamma_\redu$ of $\dirG_\redu$ such that $\Gamma$ and
  $\Gamma_\redu$ have the same bounding box. Recall that, without loss of
  generality, we assume that $\lab((u_1,u_2)) = W$.
	
  \smallskip $(\Rightarrow)$ Let $\Gamma$ be a planar drawing of $\dirG$
  having bounding box $\cal B$. By \cref{lem:lenRCCC}, we get that
  $y(u_5) < y(u_2)$ in $\Gamma$. As there is neither an $\F$-vertex nor a
  path labeled $\CRC$ in $\dirG$, we get that $\lab(u_1) = \R$ and
  $\lab(\pred(u_1))$ is either $\R$ or $\C$. If
  $\lab(\pred(u_1)) = \C$, by \cref{lem:lenCRRC}, we get that there
  exists a drawing $\Gamma'$ of $\dirG$ (which may be the same as $\Gamma$)
  whose bounding box is $\cal B$ such that $y(\pred(u_1)) = y(u_3)$ in
  $\Gamma'$. Otherwise, we get that
  $\lab((\pred(\pred(u_1)),\pred(u_1))) = \E$. So, we can apply
  \cref{lem:lenCRRC} after rotating the drawing $\Gamma$ by
  $180^{\circ}$. This, in turn, implies that $y(u_3) < y(\pred(u_1))$
  in $\Gamma$. So, without loss of generality, we can assume that
  $y(u_5) < y(u_2)$ and $y(u_3) \leq y(\pred(u_1))$.
	
  We get a drawing $\Gamma_\redu$ of $\dirG_\redu$ having the same bounding
  box $\cal B$ as follows. For every vertex
  $v \in V(\dirG_\redu) \setminus \{u'_3\}$, $x(v)$ and $y(v)$ in
  $\Gamma_\redu$ are the same as those in $\Gamma$. For $u'_3$, we assign
  $x(u'_3) = x(u_4)$ and $y(u'_3) = y(u_2)$. Observe that as
  $y(u_5) < y(u_2)$ and $y(u'_3) = y(u_2)$,
  $y(u'_3) - y(u_5) \geq 1$. Moreover, $y(u_3) \leq y(\pred(u_1))$. This implies that the weight constraints
  of the edges $(\pred(u_1),u_1)$ and $(u'_3,u_5)$ in $\dirG_\redu$
  are satisfied by $\Gamma_\redu$ and $\cal B$ is a bounding box of
  $\Gamma_\redu$. Moreover, if there does not exist any vertex
  $w \in V(\dirG)$ such that $x(w) \in [x(u_4),x(u_2)]$ and
  $y(w) \in [y(u_2),y(u_4)]$ in $\Gamma$, $\Gamma_\redu$ is a planar drawing of
  $\dirG_\redu$. Assume for contradiction there exists such a
  $w$. Observe that $\pred(w), \need(w) \notin \{u_2,u_4\}$. Then, by
  \cref{obs:LB} and \cref{lem:kittyCorner}, $\ny(u_3)$ exists and
  $(\ny(u_3),u_2)$ is a kitty corner pair of $H$, a contradiction to
  the fact that the path that we are reducing does not contain any
  kitty corner.

  \begin{figure}[tb]
  \centering
  \begin{subfigure}{0.49\textwidth}
    \centering \includegraphics[page=20]{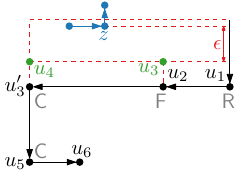}
    \caption{}
    \label{fig:proofRCCC-1}
  \end{subfigure}
  \hfill
  \begin{subfigure}{0.49\textwidth}
    \centering \includegraphics[page=21]{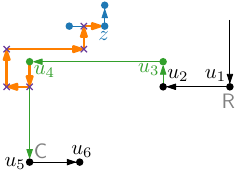}
    \caption{}
    \label{fig:proofRCCC-2}
  \end{subfigure}
  \caption{An illustration for the case considered in the reverse
    direction of the proof of \cref{lem:RCCC}. A path from $u_4$ to
    $\ny(u_1,\Gamma)$ (shown by $z$), to show that they are a pair of kitty
    corners, is drawn in orange.}
  \label{fig:proofRCCC}
\end{figure}
	
  \smallskip $(\Leftarrow)$ Let $\Gamma_\redu$ be a planar drawing of
  $\dirG_\redu$ having bounding box $\cal B$. We get a drawing $\Gamma$ of
  $\dirG$ having the same bounding box $\cal B$ as follows. For every
  vertex $v \in V(\dirG) \setminus \{u_3,u_4\}$, $x(v)$ and $y(v)$ in
  $\Gamma$ are the same as those in $\Gamma_\redu$. For $u_3$ and $u_4$, we assign
  $x(u_3) = x(u_2)$, $x(u_4) = x(u'_3)$ and
  $y(u_3) = y(u_4) = y(u_2) + \len((u_2,u_3))$. Observe that, as given in \cref{rer:RCCC}, the weight of the edge $(\pred(u_1),u_1)$ is the maximum of 
   $\len((u_2,u_3))$ and $\len((\pred(u_1), u_1))$, $\cal B$ is a bounding box of $\Gamma$.  Moreover, if there
  does not exist any vertex $w \in V(\dirG_\redu)$ such that
  $x(w) \in [x(u'_3),x(u_2)]$ and
  $y(w) \in [y(u_2),y(u_2) + \len((u_2,u_3))]$ in $\Gamma_\redu$, then $\Gamma$
  is a planar drawing of $\dirG$. Assume for contradiction there
  exists such a $w$. Observe that
  $\pred(w), \need(w) \notin \{\pred(u_1),u'_3\}$. Without loss of
  generality, for \cref{obs:LA}, we can assume that $u'_3$ is the left
  neighbor of $u_1$ as $\lab(u_2) = \F$ in $\dirG_\redu$. Then, by
  \cref{obs:LA}, $\ny(u_1)$ exists and $\lab(\ny(u_1)) = \C$. Let $\Gamma'$
  be a drawing of $\dirG$ obtained from $\Gamma_\redu$ as follows.  Let
  $\varepsilon \in \mathbb{R}$ such that $0 <\varepsilon < 1$. For
  every vertex $v \in V(\dirG) \setminus \{u_3,u_4\}$, $x(v)$ and
  $y(v)$ in $\Gamma'$ are the same as those in $\Gamma_\redu$. For $u_3$ and $u_4$,
  we assign $x(u_3) = x(u_2)$, $x(u_4) = x(u'_3)$ and
  $y(u_3) = y(u_4) = y(\ny(u_1)) - \varepsilon$. See
  \cref{fig:proofRCCC-1}. Observe that by the definition of
  $\ny(u_1)$, there exists no other vertex of $\dirG_\redu$ with
  x-coordinate in $[x(u'_3),x(u_1)]$ and y-coordinate smaller that
  that of $\ny(u_1)$ in $\Gamma_\redu$. Therefore, $\Gamma'$ is a planar drawing
  of $\dirG$.
  Similarly to the proof of
  \cref{lem:kittyCorner}, we can prove that $u_4$ and $\ny(u_1)$
  (which is also a vertex of $\dirG$) make a pair of kitty corner on
  the outer face of $H$ (see \cref{fig:proofRCCC-2}), a contradiction
  to the fact that the path that we are reducing does not contain any
  kitty corner.
\end{proof}

As the reduction rules for paths labeled $\CCCR$, $\CRRR$, and $\RRRC$ are similar, they are provided as Reduction Rules~\ref{rer:CCCR}--\ref{rer:RRRC} in the \hyperref[app:kernel]{Appendix}.

We now give the reduction rules for paths labeled $\RCCR$ or
$\CRRC$. We start by giving the reduction rule for $\RCCR$. Recall
that, due to previous reduction rules,
there is neither an $\F$-vertex nor a path labeled $\RCR$, $\CRC$,
$\RCCC$, $\CCCR$, $\CRRR$, or $\RRRC$ in $\dirG$.

\begin{reduction}\label{rer:RCCR}
  Suppose that there exists a path $P = (u_1,u_2,u_3,u_4,u_5,u_6)$ in
  $\dirG$ such that $\lab(u_2) = \lab(u_5) = \R$ and
  $\lab(u_3) = \lab(u_4) = \C$. Then, delete the vertices $u_3$ and
  $u_4$ and the edges incident to them from $\dirG$ and add a new edge
  $(u_2,u_5)$ to $\dirG$. Moreover, assign
  $\lab(u_2) = \lab(u_5) = \F$, $\len((u_2,u_5)) = \len((u_3,u_4))$
  and $\len((\pred(u_1),$ $u_1)) = \max\{\len((u_2,$
  $u_3)), \len((\pred(u_1),$ $u_1))\}$. The labels of all the
  remaining vertices and the weights of all the remaining edges stay
  the same. Let $\dirG_\redu$ be the reduced graph. See
  \cref{fig:RCCR}.
\end{reduction}

\begin{figure}[tb]
  \centering
  \begin{subfigure}{0.32\textwidth}
    \includegraphics[page=22]{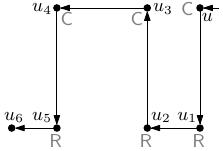}
    \caption{}
    \label{fig:RCCR-1}
  \end{subfigure}
  \hfill
  \begin{subfigure}{0.32\textwidth}
    \includegraphics[page=23]{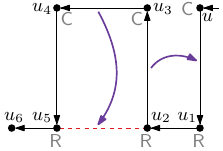}
    \caption{}
    \label{fig:RCCR-2}
  \end{subfigure}
  \hfill
  \begin{subfigure}{0.32\textwidth}
    \includegraphics[page=24]{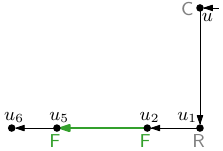}
    \caption{}
    \label{fig:RCCR-3}
  \end{subfigure}
  \caption{An illustration for \cref{rer:RCCR}. The original path $P$
    is shown in~(a) and the reduced path is shown in~(c). (b)~shows
    the projection from the edges of $\dirG$ that will be deleted when
    applying the reduction rule to the edges of $\dirG_\redu$,
    represented by curved purple edges from an edge to be deleted to
    its projected edge.}
  \label{fig:RCCR}
\end{figure}

\begin{lemma}
  \label{lem:RCCR}
  \cref{rer:RCCR} is safe.
\end{lemma}

\begin{proof}
  To prove that \cref{rer:RCCR} is safe, we need to prove that there
  exists a planar drawing $\Gamma$ of $\dirG$ if and only if there exists a
  planar drawing $\Gamma_\redu$ of $\dirG_\redu$ such that $\Gamma$ and
  $\Gamma_\redu$ have the same bounding box. Recall that, without loss of
  generality, we assume that $\lab((u_1,u_2)) = W$.
	
  \smallskip $(\Rightarrow)$ Let $\Gamma$ be a planar drawing of $\dirG$
  having bounding box $\cal B$. By \cref{lem:lenRCCR}, we get that
  $y(u_5) = y(u_2)$ in $\Gamma$. As there is neither an $\F$-vertex nor a
  path labeled $\CRC$ or $\CRRR$ in $\dirG$, we get that
  $\lab(u_1) = \R$ and $\lab(\pred(u_1)) = \C$. By \cref{lem:lenCRRC},
  we get that $y(\pred(u_1)) = y(u_3)$.  We get a drawing $\Gamma_\redu$ of
  $\dirG_\redu$ having the same bounding box $\cal B$ as follows. For
  every vertex $v \in V(\dirG_\redu)$, $x(v)$ and $y(v)$ in $\Gamma_\redu$
  are the same as those in $\Gamma$. Observe that $y(\pred(u_1)) = y(u_3)$
  implies that the weight constraint of the edge $(\pred(u_1),u_1))$
  in $\dirG_\redu$ is satisfied by $\Gamma_\redu$ and $\cal B$ is a
  bounding box of $\Gamma_\redu$. Moreover, if there does not exist any
  vertex $w \in V(\dirG)$ such that $x(w) \in [x(u_4),x(u_2)]$ and
  $y(w) \in [y(u_2),y(u_4)]$ in $\Gamma$, then $\Gamma_\redu$ is a planar
  drawing of $\dirG_\redu$. Assume for contradiction there exists such
  a $w$. Observe that $\pred(w), \need(w) \notin \{u_2,u_4\}$. Then,
  by \cref{obs:LB} and \cref{lem:kittyCorner}, $\ny(u_3)$ exists and
  $(\ny(u_3),u_2)$ is a kitty corner pair of $H$, a contradiction to
  the fact that the path that we are reducing does not contain any
  kitty corner.
	
  \smallskip $(\Leftarrow)$ Let $\Gamma_\redu$ be a planar drawing of
  $\dirG_\redu$ having bounding box $\cal B$. We get a drawing $\Gamma$ of
  $\dirG$ having the same bounding box $\cal B$ as follows. For every
  vertex $v \in V(\dirG) \setminus \{u_3,u_4\}$, $x(v)$ and $y(v)$ in
  $\Gamma$ are the same as those in $\Gamma_\redu$. For $u_3$ and $u_4$, we assign
  $x(u_3) = x(u_2)$, $x(u_4) = x(u_5)$ and
  $y(u_3) = y(u_4) = y(u_2) + \len((u_2,u_3))$. Observe that as
  $\len((\pred(u_1),u_1)) = \max\{\len((u_2,u_3)), \len((\pred(u_1),$
  $u_1))\}$, $\cal B$ is a bounding box of $\Gamma$.  Moreover, if there
  does not exist any vertex $w \in V(\dirG_\redu)$ such that
  $x(w) \in [x(u_5),x(u_2)]$ and
  $y(w) \in [y(u_2),y(u_2) + \len((u_2,u_3))]$ in $\Gamma_\redu$, then $\Gamma$
  is a planar drawing of $\dirG$. Assume for contradiction there
  exists such a $w$. Observe that
  $\pred(w), \need(w) \notin \{\pred(u_1),u_6\}$. Without loss of
  generality, from \cref{obs:LA}, we can assume that $u_6$ is the left
  neighbor of $u_1$ as $\lab(u_2) = \lab(u_5) = \F$ in
  $\dirG_\redu$. Then, by \cref{obs:LA}, $\ny(u_1)$ exists and
  $\lab(\ny(u_1)) = \C$. Let $\Gamma'$ be a drawing of $\dirG$ obtained
  from $\Gamma_\redu$ as follows.  Let $\varepsilon \in \mathbb{R}$ such
  that $0 <\varepsilon < 1$. For every vertex
  $v \in V(\dirG) \setminus \{u_3,u_4\}$, $x(v)$ and $y(v)$ in $\Gamma'$
  are the same as those in $\Gamma_\redu$. For $u_3$ and $u_4$, we assign
  $x(u_3) = x(u_2)$, $x(u_4) = x(u_5)$ and
  $y(u_3) = y(u_4) = y(\ny(u_1)) - \varepsilon$. See
  \cref{fig:proofRCCR-1}. Observe that by the definition of
  $\ny(u_1)$, there exists no other vertex of $\dirG_\redu$ with
  x-coordinate in $[x(u_5),x(u_1)]$ and y-coordinate smaller than that
  of $\ny(u_1)$ in $\Gamma_\redu$. Therefore, $\Gamma'$ is a planar drawing of
  $\dirG$.
  Similar to the proof of
  \cref{lem:kittyCorner}, we can prove that $u_4$ and $\ny(u_1)$
  (which is also a vertex of $\dirG$) make a pair of kitty corners on
  the outer face of $H$ (see \cref{fig:proofRCCR-2}), a contradiction
  to the fact that the path that we are reducing does not contain
  any kitty corner.
\end{proof}

\begin{figure}[tb]
  \centering
  \begin{subfigure}{0.48\textwidth}
    \centering \includegraphics[page=25]{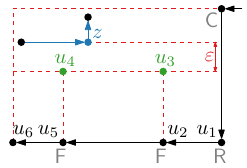}
    \caption{}
    \label{fig:proofRCCR-1}
  \end{subfigure}
  \hfill
  \begin{subfigure}{0.48\textwidth}
    \centering \includegraphics[page=26]{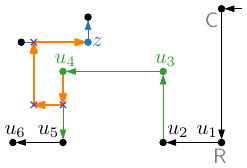}
    \caption{}
    \label{fig:proofRCCR-2}
  \end{subfigure}
  \caption{An illustration for the case considered in the reverse
    direction proof of \cref{lem:RCCR}. A path from $u_4$ to
    $\ny(u_1,\Gamma)$ (shown by $z$), to show that they are a pair of kitty
    corners, are drawn in orange.}
  \label{fig:proofRCCR}
\end{figure}

As the reduction rule for a path labeled $\CRRC$ is similar, it is provided as \cref{rer:CRRC} in the \hyperref[app:kernel]{Appendix}.

Finally, to process a path labeled $\mathsf{RRRRRRR}$, we first define
a matching $M_{RC}(\dirG)$ from the \R-vertices to \C-vertices in
$\dirG$ as follows. To define the matching, we first define the notion
of a \emph{balanced path}. Let $u$ and $v$ be two vertices of
$\dirG$. We say that $P_{u,v}$ is \emph{balanced} if the number of
\R-vertices is equal to the number of \C-vertices in $P_{u,v}$. We
have the following observation about a balanced path starting at an
\R-vertex, which will be useful in defining the matching.

\begin{observation}\label{obs:balancedPath}
  Let $P_{u,v}$ be a balanced path that starts with an
  \R-vertex. Then, there exists a \C-vertex $x$ on this path such that
  $P_{u,x}$ is balanced.
\end{observation}

We now define the matching $M_{RC}(\dirG)$. Intuitively, in
$M_{RC}(\dirG)$, every \R-vertex $u$ is matched to the closest
\C-vertex $v$ such that every internal \R-vertex on the path from $u$
to $v$ is mapped to an internal \C-vertex on the path from $u$ to
$v$. Formally, we define $M_{RC}(\dirG)$ as a matching from the set of
\R-vertices to the set of \C-vertices in $\dirG$ as follows.

An \R-vertex $u$ is matched to a \C-vertex $v$ if and only if i)
$P_{u,v}$ is balanced, and ii) for every internal \C-vertex $x$ of
$P_{u,v}$, the path $P_{u,x}$ is not balanced. Observe that as the
number of \C-vertices is larger by $4$ than the number of \R-vertices,
the matching $M_{RC}(\dirG)$ always exists. Moreover, it is also
unique. See \cref{fig:Matching}. In the following lemma, we also prove
that if an \R-vertex $u$ is matched to a \C-vertex $v$ then every
\R-vertex on the path from $u$ to $v$ is matched to a \C-vertex on the
path from $u$ to $v$.

\begin{figure}[tb]
  \centering \includegraphics[page=28]{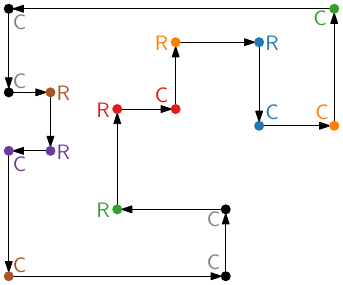}
  \caption{An illustration for the definition of the matching
    $M_{RC}(\dirG)$. The matched \R- and \C-vertex pairs are drawn in
    the same color and the remaining $4$ unmatched \C-vertices are
    drawn in black.}
  \label{fig:Matching}
\end{figure}

\begin{lemma}\label{lem:internalMatch}
  Let $u$ be an \R-vertex in $\dirG$. Moreover, let $v$ be the
  \C-vertex in $\dirG$ that is matched to $u$ by
  $M_{RC}(\dirG)$. Then, every \R-vertex on the path $P_{u,v}$ is
  matched to a \C-vertex on $P_{u,v}$ by $M_{RC}(\dirG)$.
\end{lemma}

\begin{proof}
  We prove the statement by induction on the number of vertices $\num$
  of the path $P_{u,v}$. Observe that $\num$ is even as $P_{u,v}$ is
  balanced.
	
  \smallskip
  \noindent{\bf Base case $(\num=2)$.} When $\num = 2$,
  $P_{u,v} = (u,v)$ such that $\lab(u) = \R$ and $\lab(v) = \C$. As
  $u$ is matched to $v$, the lemma is true for $\num=2$.
	
  \smallskip
  \noindent{\bf Inductive hypothesis.} Suppose that the lemma is true
  for $\num = 2t \geq 2$.
	
  \smallskip
  \noindent{\bf Inductive step.} We need to prove that the lemma is
  true for $\num=2t+2$. Let
  $P_{u,v} = (x_1=u, x_2, \ldots, x_{2t+1}, x_{2t+2} = v)$. Observe
  that $\lab(x_2) = \R$, otherwise $P_{u,x_2}$ is balanced, which
  contradicts property ii) of the definition of the matching.  We now
  look at the path $P_{x_2,x_{2t+1}}$. Since $P_{u,v}$ is balanced,
  $P_{x_2,x_{2t+1}}$ is also balanced. Moreover, $\lab(x_2) = \R$; so,
  by \cref{obs:balancedPath}, there exists a \C-vertex $x_i$, for some
  $i \in [2,2t+1]$, such that $P_{x_2,x_i}$ is balanced. Let
  $j \geq 2$ be the smallest such index for which $P_{x_2,x_j}$ is
  balanced and $\lab(x_j) = \C$. Due to the choice of $j$, for every
  internal \C-vertex $w$ of $P_{x_2,x_j}$, the path $P_{x_2,w}$ is not
  balanced. Therefore, by the definition of the matching, we get that
  $x_2$ is matched to $x_j$. As the number of vertices of
  $P_{x_2,x_j}$ is at most $2t$, from the inductive hypothesis, we get
  that every \R-vertex on $P_{x_2,x_j}$ is matched to a \C-vertex on
  $P_{x_2,x_j}$.  Now, we remove the path $P_{x_2,x_j}$ from $P_{u,v}$
  and add an arc from $u$ to $x_{j+1}$. Let $P'$ be the resulting
  path. Then, $P'$ is balanced (since we removed a balanced path), and
  property ii) is also true for $P'$ as it is true for
  $P_{u,v}$. Moreover, the number of vertices of $P'$ is at most $2t$,
  so from the inductive hypothesis, we get the every \R-vertex on $P'$
  is matched to a \C-vertex on $P'$. This, in turn, implies that every
  \R-vertex on $P_{u,v}$ is matched to a \C-vertex on $P_{u,v}$.
\end{proof}

We now consider a path $(u_1, u_2, \ldots, u_7)$ such that
$\lab(u_i) = \R$, for every $i \in [1,7]$.  Let $v_i$ be the \C-vertex
that is matched to $u_i$ by $M_{RC}(\dirG)$, for every $i \in [1,7]$.
Observe that by the property i) of the definition of $M_{RC}(\dirG)$,
the number of \R-vertices is equal to the number of \C-vertices in
$P_{u_i,v_i}$ which implies that $\rot(u_i, v_i) = -1$, for every
$i \in [1,7]$.  Moreover, considering $u = u_6$ and $v=v_6$ in \cref{lem:internalMatch}, we get that $(u_6, u_7, \ldots, v_7,\ldots, v_6)$ is a path in $\dirG$. Thus, by recursively applying the same argument to $u_5, u_4, \ldots, u_1$, we get that $(u_1, u_2, u_3, u_4, u_5, u_6, u_7, \ldots, v_7, \ldots, v_6, \ldots, v_5,\ldots, v_4, \ldots, v_3, \ldots, v_2, \ldots, v_1)$ is a path in
$\dirG$.  Therefore $\rot(u_1,v_1) = \rot(u_1,u_7)$
$+ \rot(u_7,v_7) + \rot(v_7,v_1)$ $\Rightarrow$ $\rot(v_7,v_1) = 6$.
As every \C-vertex is a reflex vertex on the outer face
$f_{\textsf{out}}$ of $\dirG$, $\rot(v_1,v_7) = -6$ on
$f_{\textsf{out}}$.  This, in turn, implies that $v_1$ and $v_7$ are a
pair of kitty corners of $f_{\textsf{out}}$.  Based on this matching
$M_{RC}(\dirG)$, we give the following counting rule to count the
vertices of a path on \R-vertices having $7$ or more vertices.

\begin{counting}\label{cou:matchingRule}
  Let $k \in \mathbb{N}$ and $t \in [0,6]$. Let
  $P = (u_1, u_2, \ldots, u_{7k+t})$ be a maximal path on \R-vertices
  in $\dirG$ not containing any kitty corner, i.e., $\lab(u_i) = \R$
  and $u_i$ is not a kitty corner of $H$, for every $i \in
  [1, 7k+t]$. Let $v_i$ be the \C-vertex matched to $u_i$ by
  $M_{RC}(\dirG)$. Then, we count the set of vertices of $P$ against
  the set of kitty corners
  $S_P = \{v_1, v_7,v_8,v_{14}, \ldots,v_{7k-6}, v_{7k}\}$.
\end{counting}

Note that, for every maximal path $P$ on \R-vertices having $7$ or
more vertices, the set $S_P$ is unique. Moreover, if two such maximal
paths $P_1$ and $P_2$ are different, then they are also
vertex-disjoint and it holds that $S_{P_1} \cap S_{P_2} =
\emptyset$. Therefore, by the above counting rule, the total number of
occurrences of \R-vertices in $\dirG$ that belong to maximal paths $P$
on \R-vertices having $7$ or more vertices and not containing any
kitty corner is bounded by a function that is linear in the number of
kitty corners of~$H$.

By applying all the reduction rules in
their respective order until they can no longer be applied, from
\cref{obs:patternMatching}, we get that $P^\trun_{c_i,c_{i+1}}$ only
contains \R-vertices and no kitty corner, for every $i \in [1,k]$. By
applying \cref{cou:matchingRule}, the number of vertices of all such
paths is bounded by a function that is linear in the number of kitty
corners of $H$. As the number of vertices of every path
$P_{c_i,c_{i+1}}$ is $10$ more than $P^\trun_{c_i,c_{i+1}}$, we get an
instance $I$ of the \textsc{Weighted OC} problem such that the number
of vertices of $I$ is bounded by a function that is linear in the
number of kitty corners of $H$. To show that it is a compression, we
also need to prove that the sizes of the edge weights of $I$ (when
encoded in binary) are bounded by a function that is linear
in~$k$. Observe that the weight of any edge of $I$ is at most $n$, the
number of vertices of $\dirG$. Therefore, if the sizes of the edge
weights of $I$ (when encoded in binary) are not bounded by a function
that is linear in~$k$, we get that $k=O(\log n)$. In this case, by
\cref{th:fpt-kitty}, we can solve the \textsc{OC} problem on $G$ in
$n^{O(1)}$ time. So, without loss of generality, we assume that
$k = \Omega(\log n)$. This, in turn, implies that the size of $I$ is
bounded by a function that is polynomial in the number of kitty
corners of~$H$.  Thus we get a polynomial compression (with a linear
number of vertices) parameterized by the number of kitty corners
of~$H$. Moreover, if $k = \Omega(\log n)$, the \textsc{Weighted OC}
problem on cycle graphs is in NP due to the fact that we can always
guess the length of each edge in the drawing (the number of guesses
and the size of each number, that is, length, to guess are polynomial
in $n$). As the \textsc{OC} problem on cycle graphs is
NP-hard~\cite{efkssw-mrpgas-CGTA22}, we can exploit the following
proposition given in the book of Fomin et
al.~\cite{fomin2019kernelization} and conclude the proof of
\cref{the:polynomial-kernel}.

\begin{proposition}[\cite{fomin2019kernelization}, Theorem~1.6]\label{pro:compToKern}
  Let $Q \subseteq \Sigma^* \times \mathbb{N}$ be a parameterized
  language, and let $R \subseteq \Sigma^*$ be a language such that the
  unparameterized version of $Q \subseteq \Sigma^* \times \mathbb{N}$
  is NP-hard and $R \subseteq \Sigma^*$ is in NP. If there is a
  polynomial compression of $Q$ into $R$, then $Q$ admits a polynomial
  kernel.
\end{proposition}

\section{Maximum Face Degree: Parameterized Hardness}
\label{se:hardness}

We show that the OC problem remains NP-hard even if all faces have
constant degrees.  Our proof elaborates on the ideas of Patrignani's
NP-hardness proof for
\textsc{OC}~\cite{DBLP:journals/comgeo/Patrignani01}, which we will recall in the following.

  Let $\phi=({\cal X},{\cal C})$ be a \textsc{SAT} instance with
  variables ${\cal X}=X_1,\ldots,X_n$ and clauses
  ${\cal C}=C_1,\ldots,C_m$. Patrignani creates a graph $G_\phi$ with
  rectilinear representation $H_\phi$ and two variables $w_\phi=9n+4$
  and $h_\phi=9m+7$ such that $H_\phi$ admits an orthogonal grid
  drawing of size $w_\phi\cdot h_\phi$ if and only if $\phi$ is
  satisfiable.  On a high level, the reduction works as follows.
	
  \begin{figure}[bh]
    \begin{subfigure}{.47\textwidth}
      \includegraphics[page=1]{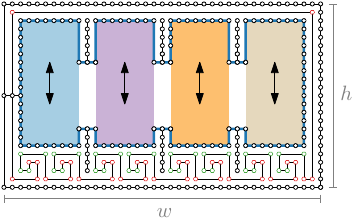}
      \caption{}
      \label{fig:app-titto-boundary-1}
    \end{subfigure}
    \hfill
    \begin{subfigure}{.47\textwidth}
      \includegraphics[page=2]{titto-boundary}
      \caption{}
      \label{fig:app-titto-boundary-2}
    \end{subfigure}
    \caption{The shifting variable rectangles (shaded) and the belt
      (the path with hexagonal vertices) in the NP-hardness proof by
      Patrignani~\cite{DBLP:journals/comgeo/Patrignani01}.}
    \label{fig:app-titto-boundary}
  \end{figure}
	
  The outer face of $H_\phi$ (the \emph{frame}) is a rectangle that
  requires width $w_\phi$ and height $h_\phi$ in its most compact
  drawing; see \cref{fig:app-titto-boundary}. Inside the frame, every
  variable is represented by a rectangular region (the \emph{variable
    rectangle}) of width~7 and height $h_\phi-7$.  Each variable
  rectangle is bounded from above and below by a horizontal path, but
  not necessarily from the left and the right, as the clause gadgets
  will be laid through them. The variable rectangles are connected
  horizontally with \emph{hinges} (short vertical paths) between them.
  Between the frame and the rectangles, there is a \emph{belt}: a long
  path of $16n+4$ vertices that consists of alternating subsequences of 4
  reflex vertices followed by 4 convex vertices.  The belt and the
  hinges make sure that every variable rectangle has to be drawn with
  exactly width~7 and height $h_\phi-7$, while the belt also ensures
  that every variable rectangle is either shifted to the top (which
  represents a \tru variable assignment) or to the bottom (which
  represents a \fal variable assignment) of the rectangle.
	
  \begin{figure}[t]
    \centering \includegraphics[page=2]{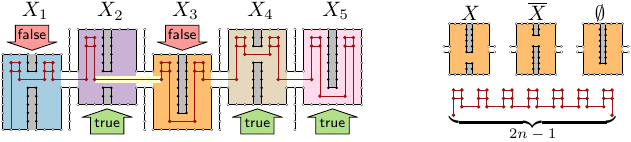}
    \caption{A clause gadget in the NP-hardness proof by
      Patrignani~\cite{DBLP:journals/comgeo/Patrignani01} for the
      clause $\overline{X_1}\vee X_2\vee \overline{X_4}$: the
      variable-clause rectangles (color shaded), the blocking
      rectangles (gray shaded), and the pathway (the path with diamond
      vertices). The segment that corresponds to a fulfilled variable
      assignment for this clause $(X_2)$ is highlighted.}
    \label{fig:app-titto-clause}
  \end{figure}
	
  Every clause is represented by a \emph{chamber} through the variable
  rectangles; see \cref{fig:app-titto-clause}. The chambers divide the
  variable rectangles evenly into $m$ subrectangles of height 9, the
  \emph{variable-clause rectangles}, with a height-2 connection
  between the rectangles; the height of the connection is ensured by
  further hinges above and below it.  Depending on the
  truth-assignments of the variables, the variable-clause rectangles
  are either shifted up or shifted down by 3.  Into each
  variable-clause rectangles, one or two \emph{blocker rectangles} of
  width~1 are inserted, centered horizontally: if the corresponding
  variable does not appear in the clause, then a rectangle of height~7
  at the top; if the variable appears negated, then a rectangle of
  height~2 at the top and a rectangle of height~5 at the bottom; and
  if the variable appears unnegated, then a rectangle of height~5 at
  the top and a rectangle of height~2 at the bottom. Into the chamber,
  a \emph{pathway} that consists of $2n-1$ \emph{$\mathsf{A}$-shapes}
  linked together by a horizontal segment is inserted.  Each of the
  $\mathsf{A}$-shapes can reside in the left or right half of a
  variable-clause rectangle. Since there are $n$ variable-clause
  rectangles, one half of such a rectangle remains empty; hence, one
  horizontal segment of the pathway has to pass a blocker rectangle,
  one half of a variable-clause rectangle, and the connection between
  two variable-clause rectangles. Thus, there must be one
  variable-clause rectangle where the opening between the blocker
  rectangles and the opening to the previous/next variable-clause
  rectangle are aligned vertically, which is exactly the case if the
  corresponding variable fulfills its truth assignment in the
  corresponding clause.  For a full reduction, see
  \cref{fig:app-titto-full}.
	
  \begin{figure}[t]
    \centering \includegraphics{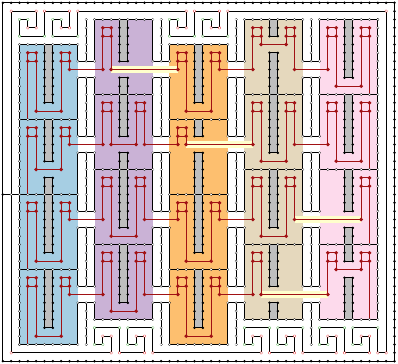}
    \caption{The full reduction by
      Patrignani~\cite{DBLP:journals/comgeo/Patrignani01} for the
      formula
      $(X_2\vee \overline{X_4})\wedge (X_1\vee X_2\vee
      \overline{X_3})\wedge (X_5)\wedge (X_4\vee \overline{X_5})$ with
      variable assignment $X_1=X_3=\fal$ and $X_2=X_4=X_5=\tru$.}
    \label{fig:app-titto-full}
  \end{figure}

\begin{theorem}
  \label{th:hard-face-degree}
  \textsc{OC} is para-NP-hard when parameterized by the maximum face
  degree.
\end{theorem}

\begin{proof}	
  To prove this theorem, we describe how to adjust Patrignani's reduction such that every face has
  constant size, thus creating a graph $G'_\phi$ that has a constant
  maximum face degree with rectilinear representation $H'_\phi$ and two
  variables $w'_\phi$ and $h'_\phi$ such that $H'_\phi$ admits an
  orthogonal grid drawing of size $w'_\phi\cdot h'_\phi$ if and only
  if $\phi$ is satisfiable.
	
  \begin{figure}[t]
    \centering \includegraphics[page=3]{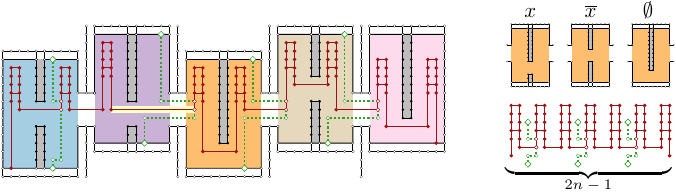}
    \caption{A clause gadget in our adjusted proof for the same clause
      and variable assignments as in \cref{fig:app-titto-clause}. The connector vertices are drawn as large (green) unfilled diamonds.}
    \label{fig:app-faces-clause}
  \end{figure}
	
  We start with the clause gadgets. Observe that the chamber consists
  of two large faces of size $O(m)$: above the pathway and below the
  pathway. To avoid these large faces, we seek to connect the
  pathway to the boundary of each of the variable-clause rectangles;
  once at the top and once at the bottom; see
  \cref{fig:app-faces-clause}. To achieve this, we first increase the
  width of all variable-clause rectangles by~2 and their height by~4.
  The height of the connection between the rectangles is increased
  to~4.  In each variable-clause rectangle, the height of the top
  blocking rectangle is increased by~3, thus increasing the height of
  the opening between/below them to~3.  The larger openings are
  required for the connections from the pathway to the boundary of the
  variable-clause rectangles. We have to make sure that the
  $\mathsf{A}$-shapes of the pathway still have to reside completely
  inside one half of a variable-clause rectangle. To achieve this, we
  also increase the necessary heights of the $\mathsf{A}$-shapes by
  adding to more vertices to the vertical segments of their top
  square. Now, the top part of each $\mathsf{A}$-shape is a rectangle
  of height at least~3, so it does not fit into any of the openings
  with size~3 or~4 without overlaps. We also add a vertex to the first
  and last vertical segment of the $\mathsf{A}$-shape that we will use
  to connect to the variable-clause rectangles later.
	
  We seek to connect every second $\mathsf{A}$-shape to the top and
  bottom boundary of a variable-clause rectangle; namely, the $2i$-th
  $\mathsf{A}$-shape of the pathway shall be connected to the
  variable-clause rectangle corresponding to variable~$X_i$,
  $1\le i\le n-1$. In particular, we always want to connect to the
  part to the right of the blocking rectangle. However, we do not know
  exactly which of the vertices to connect to, as it depends on
  whether the $\mathsf{A}$-shape is placed inside the same
  variable-clause rectangle or in the next one. Hence, we remove all
  vertices (except the corner vertices and those on the blocking
  rectangles) from the top and bottom boundary of the variable-clause
  rectangles, and add a single \emph{connector vertex} between the
  right corner and the blocking rectangle, if it exists, or between
  the right and the left corner, otherwise. To make sure that the
  variable-clause rectangles still have the required width and that
  the blocking rectangles are placed in the middle, we add a path of
  length~11 (consisting of two vertical and nine horizontal segments)
  above and below the variable-clause rectangles, and connect the two
  middle vertices of the path to the blocking rectangle (if it
  exists).
	
  Consider now the $2i$-th $\mathsf{A}$-shape and the $i$-th
  variable-clause gadget; refer again to
  \cref{fig:app-faces-clause}. We connect the first vertex of the
  $\mathsf{A}$-shape to the bottom connector vertex by a path of
  length~3 that consists of a vertical downwards segment, followed by
  a horizontal leftwards segment and a vertical downwards segment. We
  connect the second vertex of the $\mathsf{A}$-shape to the top
  connector vertex by a path of length~2 that consists of a
  horizontal leftwards segment followed by a vertical upwards
  segment.  If the $\mathsf{A}$-shape lies in the $i$-th variable
  clause gadget, then it lies in the right half of it. Since we
  increase the width of the rectangle, this half has width~3 and thus
  there is enough horizontal space for the connections if the
  $\mathsf{A}$-shape is drawn as far right as possible.  Also, since
  the opening between/below the blocking rectangles has height 2,
  there is enough vertical space between the $\mathsf{A}$-shape and
  the bottom connector if the horizontal segment to the $(2i-1)$-th
  $\mathsf{A}$-shape is drawn as far up as possible.  On the other
  hand, if the $\mathsf{A}$-shape lies in the $(i+1)$-th variable
  clause gadget, then it lies in the left half of it.  Since the
  opening between the variable-clause rectangles now has height 3,
  there is enough vertical space to fit the the connection to the
  $(2i-1)$-th $\mathsf{A}$-shape as well as the horizontal segment on
  the path to the connector vertices through the opening, and if the
  $\mathsf{A}$-shape that lies in the right half of the $i$-th
  variable-clause rectangle is drawn as far left as possible, then
  there is also enough horizontal space to fit the vertical
  connection.
	
  We still have to make sure that one half of a variable-clause
  rectangle can be ``skipped'' by the pathway if and only if the
  truth-assignment of a variable is fulfilled in the corresponding
  clause. We increase the lengths of the top hinges by~3 and the
  lengths of the bottom hinges by~1. This ensures that the opening of
  the connections has height~4 as required, and that a horizontal
  segment can pass through a connection and a blocking rectangle
  opening if and only if a variable is variable rectangle is shifted
  down (the variable is assigned \fal) and the variable appears
  negated in the clause, or if a variable is shifted up (the variable
  is assigned \fal) and the variable appears unnegated in the clause,
  therefore satisfying our requirements.  Furthermore, observe that
  each face in the clause gadget now has degree at most 55: there are
  at most 25 vertices from the pathway (2 $\mathsf{A}$-shapes + 1
  vertex), at most 18 vertices from a blocking rectangle, at most 6
  vertices from the paths to (including) the connector vertices, and
  at most 6 vertices from the boundary of the chamber.
	
  We now consider the frame and the variable rectangles. There are
  three large faces: the outer face, the face between the frame and
  the belt, and the face between the belt and the variable rectangles;
  all of these have degree $\Theta(n+m)$.
	
  \begin{figure}[t]
    \begin{subfigure}{.49\textwidth}
      \includegraphics[trim = {.8cm .6cm 0 .4cm}, clip,
      page=5]{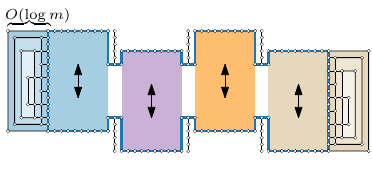}
      \caption{}
      \label{fig:app-faces-boxes-1}
    \end{subfigure}
    \hfill
    \begin{subfigure}{.49\textwidth}
      \centering \includegraphics[scale=.7,page=13]{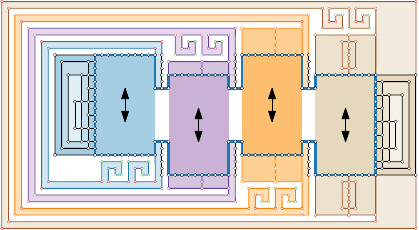}
      \caption{}
      \label{fig:app-faces-boxes-2}
    \end{subfigure}
    \caption{The frame in our adjusted NP-hardness reduction. (a) The
      left and right extensions of the first and last variable
      rectangle; (b) the belts around the variable rectangles.}
    \label{fig:app-faces-boxes}
  \end{figure}
	
  We first show how to deal with the face between the variable
  rectangles and the belt. We ignore the belt for now and focus on the
  path around the variable rectangles; see
  \cref{fig:app-faces-boxes}. This path consists of two vertical paths of
  length $O(m)$ on the left and the right of the first and the last
  variable rectangle, respectively, and an $x$-monotone path of length
  $O(n)$ between them.  Consider the vertical path on the left of the
  first variable rectangle, let $z\in O(m)$ be its length and pick
  some constant integer $c\ge 2$.  We first place a rectangle to the
  left of the path by connecting the topmost with the bottommost vertex
  with a path of length~3, consisting of a horizontal leftward segment
  followed by a vertical downward and a horizontal rightward segment.
  This creates a new interior rectangular face with 2 vertices on the
  left and~$z$ vertices on the right.  We now iteratively decrease the
  degree of this face, adding new faces of degree~10 and rectangular
  faces of degree at most~$c+3$, until this face also has constant
  degree.  We place at most~$z/c$ rectangles to the left of the right
  path, each of height at most~$c$, by adding a vertical path of at
  most~$z/c+1$ vertices that are connected to the second vertex from
  the top, the second vertex from the bottom, and every $c$-th vertex
  in between. This way, we create at most~$z/c$ new faces of degree at
  most $c+3$, and a new vertical path of length at most~$z/c+1$. We
  connect the topmost and the bottommost vertex of this new path to
  create a new rectangular face as before. This creates a C-shaped
  face of degree~10, and a new interior rectangular face with two
  vertices on the left and at most~$z/c+1$ vertices on the right.  We
  repeat this process at most~$\log_c(z)+1$ times, until the vertical
  path on the right has length at most~$c$ and thus the interior face
  has degree at most~$c+6$; see \cref{fig:app-faces-boxes-1}. We do
  the same to the vertical path on the right of the last variable
  rectangle. This way, we create at most~$2z$ new faces, all of constant
  degree, and the full path now has length $O(n)$.
	
  Instead of placing a single belt and outer face, we now add a small
  belt of only 20 vertices around every variable rectangle. Namely,
  the belt starts and ends at the rightmost vertex of the top and
  bottom boundary of the variable rectangles, respectively. For the
  first variable rectangle, the belt bounds a face of at most $30+c$
  vertices: 20 from the belt, 10 from the top, and $c$ from the
  extension to the left.  We add another path between the extreme
  vertices of the hinges between the variable rectangles that consists
  of four $90^\circ$-vertices; this creates a face of degree~40: 4
  from this path, 20 from the belt, 10 from the hinges, and 6 from the
  variable rectangles. This path also functions as a bounding box for
  the belt: the minimum bounding box for this path can only be
  achieved if both spirals of the belt are either below or above the
  variable rectangle, thus creating the binary choice of shifting the
  variable rectangle up (\tru) or down (\fal).
	
  For the next variable rectangles, we have to be a bit more
  careful. Repeating the process as for the first variable rectangle
  does not immediately work, as the new belt has to walk around the
  path that describes the bounding box for the previous variable
  rectangle plus its belt. Thus, there would be a gap of height~2
  above and below the variable rectangle. This means that the variable
  rectangle could shift up- and downwards by two coordinates, which
  makes the clause gadgets invalid.  To avoid this issue and to close
  the gap, we place a box of height 2 above and below the second
  variable rectangle, and we place a box of height $2(i-1)$ around the
  $i$-th variable rectangle, for each $i \in \{3,\dots,n\}$.  To avoid
  a face of degree $O(n)$, we use the same strategy as for the
  leftmost path of the left variable gadget: we place a rectangle with
  an interior vertical path of height $2(i-1)$ and we keep refining
  the resulting faces until they all have constant degree. This way,
  we create $2(n-1)$ faces of degree at most 46: 20 from the belt, 4
  from the bounding box path, 10 from the hinges, 6 from the vertical
  extension boxes, and 6 from the variable rectangles.  The outer face
  of the graph has only degree~6: four from the bounding box, and the
  two corners of the rightwards extension of the last variable
  rectangle.
	
  With this adjustment, we can still encode the variable assignments
  of a variable rectangle being either shifted up (\tru) or down
  (\fal), and we can draw the rectilinear representation with minimum
  width and height if and only if every clause gadget has a compact
  drawing, thus the whole rectilinear representation has an orthogonal
  drawing with area $w'_\phi\times h'_\phi$ for some fixed width
  $w'_\phi$ and height $h'_\phi$ if and only if $\phi$ is satisfyable.
  Clearly, $w'_\phi$ and $h'_\phi$ (and the number of vertices) are
  polynomial in $n$ and $m$, and as described above all faces have
  constant degree.
	
  Hence, \textsc{OC} remains NP-hard even for graphs of constant face
  degree and thus it is para-NP-hard when parameterized by the maximum
  face degree.
\end{proof}

\section{Height of the Representation: An XP Algorithm}
\label{se:xp-height}

Given a connected planar graph~$G$ of vertex-degree at most~4, recall that
the \emph{height} of a rectilinear representation~$H$
of~$G$ is the 
minimum height of a rectangular section of the integer grid required
to draw~$H$ orthogonally. By ``guessing'' for every column of the
drawing what lies on each grid point, we obtain an XP algorithm for
\textsc{OC} parameterized by the height of the representation.

\begin{theorem}
  \label{th:xp-height}
  Given a connected planar graph~$G$, a rectilinear representation~$H$
  of~$G$ with $n$ vertices, and an integer $h \ge 1$, we can decide,
  in $n^{O(h)}$ time, whether $H$ admits an orthogonal drawing of
  height~$h$.  In other words, \textsc{OC} is XP with respect
  to~$h$.
\end{theorem}

\begin{proof}
  We want to decide whether $H$ admits an orthogonal drawing on a grid
  that consists of~$h$ horizontal lines.  Given a solution, that is, a
  drawing of~$H$, we can remove any column that does not contain any
  vertex.  Hence it suffices to check whether there
  exists a drawing of~$H$ on a grid of (height~$h$ and) width
  $w \le n$.
	
  To this end, we use dynamic programming (DP) with a table~$B$.  Each
  entry of~$B[c,t]$ corresponds to a column $c$ of the grid and an
  $h$-tuple~$t$.  For an example, see \cref{fig:dp-table}.  Each
  component of~$t$ represents an object (if any) that lies on the
  corresponding grid point in column~$c$.  In a drawing of~$H$, a grid
  point~$g$ can either be empty or it is occupied by a vertex or by an
  edge.  In the case of an edge, the edge can pass through~$g$ either
  horizontally or vertically, 
  as prescribed by~$H$.  Let $\mathcal{T}$ be the set of $h$-tuples
  constructed in this way.  Due to our observation regarding~$w$
  above, $\mathcal{T}$ does not contain any $h$-tuple that consists
  exclusively of horizontally crossing edges and empty grid points.
  Hence $|\mathcal{T}| \in (O(n))^h$.

  \begin{figure}[tb]
    \begin{subfigure}[b]{.66\linewidth}
      \centering
      \includegraphics[page=1]{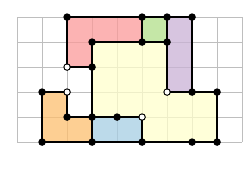}
      \includegraphics[page=3]{compaction-new}
      \caption{two drawings of the same orthogonal representation}
      \label{fig:ortho-rep}
    \end{subfigure}
    \hfill
    \begin{subfigure}[b]{.32\linewidth}
      \centering
      \includegraphics{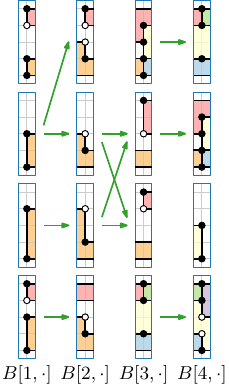}
      \caption{part of the DP table} 
      \label{fig:dp-table}
    \end{subfigure}

    \caption{Part of the table~(b) used by our DP for checking whether
      the orthogonal representation depicted in~(a) admits a drawing
      of height~5.  The arcs that enter an entry $B[c,t]$ come from
      the entries that are taken into account when computing
      $B[c,t]$.  The white vertices are part of a pair of kitty
      corners.}
  \end{figure}
	
  Observe that, in an orthogonal drawing, each column can be
  considered a \emph{cut} of the graph that the drawing represents.
  Hence, any $t \in \mathcal{T}$ uniquely defines all elements of~$H$
  that, in any height-$h$ grid drawing, cross~$c$ or lie to the left
  of~$c$.  We refer to this part of~$H$ as $\lef_H(t)$.  Analogously,
  we define the set $\on_H(t)$ as the elements that are encoded by~$t$,
  that is, vertices or edges of~$H$ that lie on the corresponding grid
  column, or edges that cross the column.  Note that every horizontal
  edge in~$\on_H(t)$ is also contained in~$\lef_H(t)$, whereas every
  vertical edge in~$\on_H(t)$ is {\em not} contained in~$\lef_H(t)$.
  The table entry $B[c,t]$ stores a Boolean value that is true if an
  orthogonal drawing of $\lef_H(t)$ on a grid of size $c \times h$
  (counting grid lines) exists, false otherwise.
	
  We fill the table $B[\cdot,\cdot]$ from left to right,
  that is, in order of increasing first component.  If
  $c=1$, we check, for each $t \in \mathcal{T}$, whether
  $\lef_H(t)=\emptyset$.  In this case, we set $B[c,t]=\tru$,
  otherwise to \fal.  For the case $c>1$, we define~$t' \in
  \mathcal{T}$ to be a \emph{direct predecessor} of~$t \in\mathcal{T}$
  if $\lef_H(t)=\lef_H(t') \cup \on_H(t')$ and, if we place $t'$
  on a column directly in front of~$t$, they ``match'', that is, every
  edge that starts in~$t'$ or crosses~$t'$ either ends on~$t$ or
  crosses~$t$ (in the same row in~$t$ and~$t'$).
  In particular, edges that cross both columns must lie in the same row.
  For every $t \in \mathcal{T}$, we set $B[c,t]=\tru$ if $t$ has a
  direct predecessor $t'$ such that $B[c-1,t']=\tru$.  Otherwise, we
  set $B[c,t]=\fal$.
	
  The DP returns \tru if and only if there exists a pair
  $(c,t) \in [w] \times \mathcal{T}$ such that $B[c,t]=\tru$ and
  $\lef_H(t) \cup \on_H(t) = V(H) \cup E(H)$.
	
  The number of different $h$-tuples is $(O(n))^h$.  For each
  $(c,t) \in [w] \times \mathcal{T}$, we can compute
  $B[c,t]$ in $O(h) \cdot (O(n))^h$ time.  Hence, the DP runs in
  $n^{O(h)}$ total time.
\end{proof}

\section{Conclusions and Open Problems}
\label{se:conclusions}

In this paper we studied the problem of computing an orthogonal
drawing of a given orthogonal representation such that the area of the
drawing is minimized. While the problem is known to be NP-complete in
general, we show that it is FPT when parameterized by the number of
kitty corners and that, if the graph is a simple cycle, there is a
kernel of polynomial size.  We also considered parameters other than
the number of kitty corners: We proved that it remains NP-hard when
parameterized by the maximum face degree, and showed the existence of
an XP algorithm when the parameter is the height of the  drawing. We
conclude by listing some open problems that, in our opinion, can be
the subject of further investigation.

\begin{enumerate}[label=(\arabic*)]
\item Can we find a polynomial kernel for \textsc{OC} with respect to
  the number of kitty corners, or at least with respect to the number
  of kitty corners plus the number of faces, for general graphs? 
\item Does \textsc{OC} admit an FPT algorithm parameterized by the
  height of the orthogonal representation?
\item Is \textsc{OC} solvable in $2^{O(\sqrt{n})}$ time?  This bound
  would be tight assuming that the Exponential Time Hypothesis
  is true.
\item If we parameterize by the number of {\em pairs} of kitty
  corners, can we achieve substantially better running times?
\item Can we substantially reduce the base of the exponent in the running time of the FPT algorithm? In particular, is the problem solvable in time $c^k\cdot n^{O(1)}$ for $c<10$?
\end{enumerate}

\noindent{\bf Acknowledgements.} This research was initiated at
Dagstuhl Seminar 21293: {\em Parameterized Complexity in Graph
  Drawing}.  We are grateful to the organizers for making the seminar
possible and to the other participants for useful discussions.

\bibliographystyle{elsarticle-num} 
\bibliography{compaction}

\begin{thebibliography}{10}
\expandafter\ifx\csname url\endcsname\relax
  \def\url#1{\texttt{#1}}\fi
\expandafter\ifx\csname urlprefix\endcsname\relax\def\urlprefix{URL }\fi
\expandafter\ifx\csname href\endcsname\relax
  \def\href#1#2{#2} \def\path#1{#1}\fi

\bibitem{DBLP:books/ph/BattistaETT99}
G.~{Di Battista}, P.~Eades, R.~Tamassia, I.~G. Tollis, Graph Drawing:
  Algorithms for the Visualization of Graphs, Prentice-Hall, 1999.

\bibitem{DBLP:conf/dagstuhl/1999dg}
M.~Kaufmann, D.~Wagner (Eds.), Drawing Graphs, Methods and Models, Vol. 2025 of
  Lecture Notes in Computer Science, Springer, 2001.
\newblock \href {https://doi.org/10.1007/3-540-44969-8}
  {\path{doi:10.1007/3-540-44969-8}}.

\bibitem{DBLP:journals/comgeo/BridgemanBDLTV00}
S.~S. Bridgeman, G.~{Di Battista}, W.~Didimo, G.~Liotta, R.~Tamassia,
  L.~Vismara, Turn-regularity and optimal area drawings of orthogonal
  representations, Comput. Geom. 16~(1) (2000) 53--93.
\newblock \href {https://doi.org/10.1016/S0925-7721(99)00054-1}
  {\path{doi:10.1016/S0925-7721(99)00054-1}}.

\bibitem{DBLP:journals/comgeo/Patrignani01}
M.~Patrignani, On the complexity of orthogonal compaction, Comput. Geom. 19~(1)
  (2001) 47--67.
\newblock \href {https://doi.org/10.1016/S0925-7721(01)00010-4}
  {\path{doi:10.1016/S0925-7721(01)00010-4}}.

\bibitem{efkssw-mrpgas-CGTA22}
W.~S. Evans, K.~Fleszar, P.~Kindermann, N.~Saeedi, C.-S. Shin, A.~Wolff,
  Minimum rectilinear polygons for given angle sequences, Comput. Geom.
  100~(101820) (2022) 1--39.
\newblock \href {https://doi.org/10.1016/j.comgeo.2021.101820}
  {\path{doi:10.1016/j.comgeo.2021.101820}}.

\bibitem{bes-ioc-JGAA12}
M.~J. Bannister, D.~Eppstein, J.~A. Simons, Inapproximability of orthogonal
  compaction, J. Graph Algorithms Appl. 16~(3) (2012) 651--673.
\newblock \href {https://doi.org/10.7155/jgaa.00263}
  {\path{doi:10.7155/jgaa.00263}}.

\bibitem{DBLP:journals/dagstuhl-reports/GanianMNZ21}
R.~Ganian, F.~Montecchiani, M.~N{\"{o}}llenburg, M.~Zehavi, Parameterized
  complexity in graph drawing ({Dagstuhl Seminar} 21293), Dagstuhl Reports
  11~(6) (2021) 82--123.
\newblock \href {https://doi.org/10.4230/DagRep.11.6.82}
  {\path{doi:10.4230/DagRep.11.6.82}}.

\bibitem{DBLP:journals/jgaa/BannisterE18}
M.~J. Bannister, D.~Eppstein, Crossing minimization for 1-page and 2-page
  drawings of graphs with bounded treewidth, J. Graph Algorithms Appl. 22~(4)
  (2018) 577--606.
\newblock \href {https://doi.org/10.7155/jgaa.00479}
  {\path{doi:10.7155/jgaa.00479}}.

\bibitem{DBLP:journals/jgaa/BhoreGMN20}
S.~Bhore, R.~Ganian, F.~Montecchiani, M.~N{\"{o}}llenburg, Parameterized
  algorithms for book embedding problems, J. Graph Algorithms Appl. 24~(4)
  (2020) 603--620.
\newblock \href {https://doi.org/10.7155/jgaa.00526}
  {\path{doi:10.7155/jgaa.00526}}.

\bibitem{DBLP:journals/jgaa/BhoreGMN22}
S.~Bhore, R.~Ganian, F.~Montecchiani, M.~N{\"{o}}llenburg, Parameterized
  algorithms for queue layouts, J. Graph Algorithms Appl. 26~(3) (2022)
  335--352.
\newblock \href {https://doi.org/10.7155/JGAA.00597}
  {\path{doi:10.7155/JGAA.00597}}.

\bibitem{DBLP:conf/compgeom/BinucciLGDMP19}
C.~Binucci, G.~{Da Lozzo}, E.~{Di Giacomo}, W.~Didimo, T.~Mchedlidze,
  M.~Patrignani, Upward book embeddings of st-graphs, in: G.~Barequet, Y.~Wang
  (Eds.), Symp. Comput. Geom. (SoCG), Vol. 129 of LIPIcs, Schloss Dagstuhl~--
  Leibniz-Zentrum f{\"{u}}r Informatik, 2019, pp. 13:1--13:22.
\newblock \href {https://doi.org/10.4230/LIPIcs.SoCG.2019.13}
  {\path{doi:10.4230/LIPIcs.SoCG.2019.13}}.

\bibitem{kobayashi_et_al:LIPIcs:2018:8566}
Y.~Kobayashi, H.~Ohtsuka, H.~Tamaki, An improved fixed-parameter algorithm for
  one-page crossing minimization, in: D.~Lokshtanov, N.~Nishimura (Eds.), Int.
  Symp. Parametr. \& Exact Comput. (IPEC), Vol.~89 of LIPIcs, Schloss
  Dagstuhl~-- Leibniz-Zentrum f\"ur Informatik, 2018, pp. 25:1--25:12.
\newblock \href {https://doi.org/10.4230/LIPIcs.IPEC.2017.25}
  {\path{doi:10.4230/LIPIcs.IPEC.2017.25}}.

\bibitem{cdfgrs-paup-SoCG22}
S.~Chaplick, E.~Di~Giacomo, F.~Frati, R.~Ganian, C.~N. Raftopoulou, K.~Simonov,
  Parameterized algorithms for upward planarity, in: X.~Goaoc, M.~Kerber
  (Eds.), Symp. Comput. Geom. (SoCG), Vol. 224 of LIPIcs, Schloss Dagstuhl~--
  Leibniz-Zentrum f{\"u}r Informatik, 2022, pp. 26:1--26:16.
\newblock \href {https://doi.org/10.4230/LIPIcs.SoCG.2022.26}
  {\path{doi:10.4230/LIPIcs.SoCG.2022.26}}.

\bibitem{DBLP:journals/jcss/GiacomoLM22}
E.~{Di Giacomo}, G.~Liotta, F.~Montecchiani, Orthogonal planarity testing of
  bounded treewidth graphs, J. Comput. Syst. Sci. 125 (2022) 129--148.
\newblock \href {https://doi.org/10.1016/j.jcss.2021.11.004}
  {\path{doi:10.1016/j.jcss.2021.11.004}}.

\bibitem{DBLP:conf/isaac/0002SZ21}
S.~Gupta, G.~Sa'ar, M.~Zehavi, Grid recognition: Classical and parameterized
  computational perspectives, in: H.-K. Ahn, K.~Sadakane (Eds.), Int. Symp.
  Algorithms \& Comput. (ISAAC), Vol. 212 of LIPIcs, Schloss Dagstuhl~--
  Leibniz-Zentrum f{\"{u}}r Informatik, 2021, pp. 37:1--37:15.
\newblock \href {https://doi.org/10.4230/LIPIcs.ISAAC.2021.37}
  {\path{doi:10.4230/LIPIcs.ISAAC.2021.37}}.

\bibitem{DBLP:conf/wg/LozzoEG018}
G.~{Da Lozzo}, D.~Eppstein, M.~T. Goodrich, S.~Gupta, Subexponential-time and
  {FPT} algorithms for embedded flat clustered planarity, in:
  A.~Brandst{\"a}dt, E.~K{\"o}hler, K.~Meer (Eds.), Int. Workshop
  Graph-Theoretic Concepts Comput. Sci. (WG), Vol. 11159 of LNCS, Springer,
  2018, pp. 111--124.
\newblock \href {https://doi.org/10.1007/978-3-030-00256-5_10}
  {\path{doi:10.1007/978-3-030-00256-5_10}}.

\bibitem{LiRuTa-PCGPRCO-22}
G.~Liotta, I.~Rutter, A.~Tappini, Parameterized complexity of graph planarity
  with restricted cyclic orders, in: M.~A. Bekos, M.~Kaufmann (Eds.), Int.
  Workshop on Graph-Theoretic Concepts in Comput. Sci. (WG), Vol. 13453 of
  LNCS, Springer, 2022, pp. 383--397.
\newblock \href {https://doi.org/10.1007/978-3-031-15914-5_28}
  {\path{doi:10.1007/978-3-031-15914-5_28}}.

\bibitem{DBLP:journals/algorithmica/LozzoEGG21}
G.~{Da Lozzo}, D.~Eppstein, M.~T. Goodrich, S.~Gupta, C-planarity testing of
  embedded clustered graphs with bounded dual carving-width, Algorithmica
  83~(8) (2021) 2471--2502.
\newblock \href {https://doi.org/10.1007/s00453-021-00839-2}
  {\path{doi:10.1007/s00453-021-00839-2}}.

\bibitem{DBLP:journals/jgaa/BannisterCE18}
M.~J. Bannister, S.~Cabello, D.~Eppstein, Parameterized complexity of
  1-planarity, J. Graph Algorithms Appl. 22~(1) (2018) 23--49.
\newblock \href {https://doi.org/10.7155/jgaa.00457}
  {\path{doi:10.7155/jgaa.00457}}.

\bibitem{DBLP:conf/gd/BannisterES13}
M.~J. Bannister, D.~Eppstein, J.~A. Simons, Fixed parameter tractability of
  crossing minimization of almost-trees, in: S.~Wismath, A.~Wolff (Eds.), Int.
  Symp. Graph Drawing (GD), Vol. 8242 of LNCS, Springer, 2013, pp. 340--351.
\newblock \href {https://doi.org/10.1007/978-3-319-03841-4_30}
  {\path{doi:10.1007/978-3-319-03841-4_30}}.

\bibitem{DBLP:journals/algorithmica/DujmovicFKLMNRRWW08}
V.~Dujmovi{\'c}, M.~R. Fellows, M.~Kitching, G.~Liotta, C.~McCartin,
  N.~Nishimura, P.~Ragde, F.~A. Rosamond, S.~Whitesides, D.~R. Wood, On the
  parameterized complexity of layered graph drawing, Algorithmica 52~(2) (2008)
  267--292.
\newblock \href {https://doi.org/10.1007/s00453-007-9151-1}
  {\path{doi:10.1007/s00453-007-9151-1}}.

\bibitem{DBLP:journals/jda/DujmovicFK08}
V.~Dujmovi{\'c}, H.~Fernau, M.~Kaufmann, Fixed parameter algorithms for
  one-sided crossing minimization revisited, J. Discrete Algorithms 6~(2)
  (2008) 313--323.
\newblock \href {https://doi.org/10.1016/j.jda.2006.12.008}
  {\path{doi:10.1016/j.jda.2006.12.008}}.

\bibitem{cdggkw-pcsn-GD23}
S.~Cornelsen, G.~{Da Lozzo}, L.~Grilli, S.~Gupta, J.~Kratochv\'{\i}l, A.~Wolff,
  The parametrized complexity of the segment number, in: M.~Bekos, M.~Chimani
  (Eds.), Int. Symp. Graph Drawing \& Network Vis. (GD), Vol. 14466 of LNCS,
  Springer, 2023, pp. 97--113.
\newblock \href {https://doi.org/10.1007/978-3-031-49275-4_7}
  {\path{doi:10.1007/978-3-031-49275-4_7}}.

\bibitem{cflrvw-cdgfl-JGAA23}
S.~Chaplick, K.~Fleszar, F.~Lipp, A.~Ravsky, O.~Verbitsky, A.~Wolff, The
  complexity of drawing graphs on few lines and few planes, J. Graph Algorithms
  Appl. 27~(6) (2023) 459--488.
\newblock \href {https://doi.org/10.7155/jgaa.00630}
  {\path{doi:10.7155/jgaa.00630}}.

\bibitem{bcghvw-bcong-SIDMA24}
M.~Balko, S.~Chaplick, R.~Ganian, S.~Gupta, M.~Hoffmann, P.~Valtr, A.~Wolff,
  Bounding and computing obstacle numbers of graphs, SIAM J. Discrete Math.
  38~(2) (2024) 1537--1565.
\newblock \href {https://doi.org/10.1137/23M1585088}
  {\path{doi:10.1137/23M1585088}}.

\bibitem{DBLP:journals/dcg/Biedl11}
T.~C. Biedl, Small drawings of outerplanar graphs, series-parallel graphs, and
  other planar graphs, Discret. Comput. Geom. 45~(1) (2011) 141--160.
\newblock \href {https://doi.org/10.1007/s00454-010-9310-z}
  {\path{doi:10.1007/s00454-010-9310-z}}.

\bibitem{cflrvw-dgflf-JoCG20}
S.~Chaplick, K.~Fleszar, F.~Lipp, A.~Ravsky, O.~Verbitsky, A.~Wolff, Drawing
  graphs on few lines and few planes, J. Comput. Geom. 11~(1) (2020) 433--475.
\newblock \href {https://doi.org/10.20382/jocg.v11i1a17}
  {\path{doi:10.20382/jocg.v11i1a17}}.

\bibitem{DBLP:books/sp/CyganFKLMPPS15}
M.~Cygan, F.~V. Fomin, L.~Kowalik, D.~Lokshtanov, D.~Marx, M.~Pilipczuk,
  M.~Pilipczuk, S.~Saurabh, Parameterized Algorithms, Springer, 2015.
\newblock \href {https://doi.org/10.1007/978-3-319-21275-3}
  {\path{doi:10.1007/978-3-319-21275-3}}.

\bibitem{downey2013fundamentals}
R.~G. Downey, M.~R. Fellows, Fundamentals of Parameterized Complexity, Vol.~4
  of TCS, Springer, 2013.
\newblock \href {https://doi.org/10.1007/978-1-4471-5559-1}
  {\path{doi:10.1007/978-1-4471-5559-1}}.

\bibitem{fomin2019kernelization}
F.~V. Fomin, D.~Lokshtanov, S.~Saurabh, M.~Zehavi, Kernelization: Theory of
  Parameterized Preprocessing, Cambridge University Press, 2019.

\bibitem{DBLP:journals/tcs/BattistaT88}
G.~{Di Battista}, R.~Tamassia, Algorithms for plane representations of acyclic
  digraphs, Theor. Comput. Sci. 61 (1988) 175--198.
\newblock \href {https://doi.org/10.1016/0304-3975(88)90123-5}
  {\path{doi:10.1016/0304-3975(88)90123-5}}.

\bibitem{DBLP:journals/algorithmica/BertolazziBLM94}
P.~Bertolazzi, G.~{Di Battista}, G.~Liotta, C.~Mannino, Upward drawings of
  triconnected digraphs, Algorithmica 12~(6) (1994) 476--497.
\newblock \href {https://doi.org/10.1007/BF01188716}
  {\path{doi:10.1007/BF01188716}}.

\bibitem{p-2009}
C.~A. Pickover, The Math Book, Sterling, 2009.

\bibitem{bo-79}
J.~L. Bentley, T.~Ottmann, Algorithms for reporting and counting geometric
  intersections, IEEE Trans. Comput. C-28~(9) (1979) 643--647.
\newblock \href {https://doi.org/10.1109/TC.1979.1675432}
  {\path{doi:10.1109/TC.1979.1675432}}.

\bibitem{sh-76}
M.~I. Shamos, D.~Hoey, Geometric intersection problems, in: Symp. Foundat.
  Comput. Sci. (FoCS), 1976, pp. 208--215.
\newblock \href {https://doi.org/10.1109/SFCS.1976.16}
  {\path{doi:10.1109/SFCS.1976.16}}.

\bibitem{t-eggmn-87}
R.~Tamassia, On embedding a graph in the grid with the minimum number of bends,
  SIAM J. Comput. 16~(3) (1987) 421--444.
\newblock \href {https://doi.org/10.1137/0216030} {\path{doi:10.1137/0216030}}.

\end{thebibliography}

\newpage

\appendix

\section*{Appendix}
\label{app:kernel}

In this section, we provide the observations, reduction rules, and lemmas omitted in \cref{sec:kernel} to improve readability. 
We first give observations similar to \cref{obs:LB} for the cases when either $\rig(v)$ and $\bel(v)$ (\cref{fig:nearXnY-2}), or $\rig(v)$ and $\abo(v)$ (\cref{fig:nearXnY-3}), or $\lef(v)$ and $\abo(v)$ (\cref{fig:nearXnY-4}) exist.

\begin{observation}\label{obs:RB}
  Let $v$ be an \R- or a \C-vertex of $\dirG$ such that $\rig(v)$ and
  $\bel(v)$ exist. Let $\Gamma$ be a drawing of $\dirG$ such that there
  exists a vertex $u$ for which $(i)$
  $x(u) \in [x(\bel(v)),x(\rig(v))]$, $(ii)$
  $y(u) \in [y(\bel(v)), y(\rig(v))]$, and $(iii)$
  $u, \pred(u), \need(u) \notin \{\rig(v), \bel(v)\}$. Let $S$ be the
  set of all such vertices. Let $\maY = \max\{y(v) \mid v \in S\}$ and
  $\miX = \min\{x(v) \mid v \in S\}$.
	
  Then, $\nx(v,\Gamma)$ is the vertex in $S$ such that $x(\nx(v,\Gamma)) = \miX$
  and $y(\nx(v,\Gamma)) = \max\{y(v) \mid v \in S \wedge x(v) =
  \miX\}$. Similarly, $\ny(v,\Gamma)$ is the vertex in $S$ such that
  $y(\ny(v,\Gamma)) = \maY$ and
  $x(\ny(v,\Gamma)) = \min\{x(v) \mid v \in S \wedge y(v) = \maY\}$. It
  follows from the definition that the neighbors of $\nx(v,\Gamma)$ (resp.,
  $\ny(v,\Gamma)$) are to the right of and below $\nx(v,\Gamma)$ (resp.,
  $\ny(v,\Gamma)$). Moreover, from the planarity of $\Gamma$, it follows that
  $\lab(\nx(v,\Gamma)) \in \{\R,\C\} \setminus \{\lab(v)\}$ and
  $\lab(\ny(v,\Gamma)) \in \{\R,\C\} \setminus \{\lab(v)\}$. See
  \cref{fig:nearXnY-2}.
\end{observation}

\begin{observation}\label{obs:RA}
  Let $v$ be an \R- or a \C-vertex of $\dirG$ such that $\rig(v)$ and
  $\abo(v)$ exist. Let $\Gamma$ be a drawing of $\dirG$ such that there
  exists a vertex $u$ for which $(i)$
  $x(u) \in [x(\abo(v)),x(\rig(v))]$, $(ii)$
  $y(u) \in [y(\rig(v)), y(\abo(v))]$, and $(iii)$
  $u, \pred(u), \need(u) \notin \{\rig(v), \abo(v)\}$. Let $S$ be the
  set of all such vertices. Let $\miY = \min\{y(v) \mid v \in S\}$ and
  $\miX = \min\{x(v) \mid v \in S\}$.
	
  Then, $\nx(v,\Gamma)$ is the vertex in $S$ such that $x(\nx(v,\Gamma)) = \miX$
  and $y(\nx(v,\Gamma)) = \min\{y(v) \mid v \in S \wedge x(v) =
  \miX\}$. Similarly, $\ny(v,\Gamma)$ is the vertex in $S$ such that
  $y(\ny(v,\Gamma)) = \miY$ and
  $x(\ny(v,\Gamma)) = \min\{x(v) \mid v \in S \wedge y(v) = \miY\}$. It
  follows from the definition that the neighbors of $\nx(v,\Gamma)$ (resp.,
  $\ny(v,\Gamma)$) are to the right of and above $\nx(v,\Gamma)$ (resp.,
  $\ny(v,\Gamma)$). Moreover, from the planarity of $\Gamma$, it follows that
  $\lab(\nx(v,\Gamma)) \in \{\R,\C\} \setminus \{\lab(v)\}$ and
  $\lab(\ny(v,\Gamma)) \in \{\R,\C\} \setminus \{\lab(v)\}$. See
  \cref{fig:nearXnY-3}.
\end{observation}

\begin{observation}\label{obs:LA}
  Let $v$ be an \R- or a \C-vertex of $\dirG$ such that $\lef(v)$ and
  $\abo(v)$ exist. Let $\Gamma$ be a drawing of $\dirG$ such that there
  exists a vertex $u$ for which $(i)$
  $x(u) \in [x(\lef(v)),x(\abo(v))]$, $(ii)$
  $y(u) \in [y(\lef(v)), y(\abo(v))]$, and $(iii)$
  $u, \pred(u), \need(u) \notin \{\lef(v), \abo(v)\}$. Let $S$ be the
  set of all such vertices. Let $\miY = \max\{y(v) \mid v \in S\}$ and
  $\maX = \max\{x(v) \mid v \in S\}$.
	
  Then, $\nx(v,\Gamma)$ is the vertex in $S$ such that $x(\nx(v,\Gamma)) = \maX$
  and $y(\nx(v,\Gamma)) = \min\{y(v) \mid v \in S \wedge x(v) =
  \maX\}$. Similarly, $\ny(v,\Gamma)$ is the vertex in $S$ such that
  $y(\ny(v,\Gamma)) = \miY$ and
  $x(\ny(v,\Gamma)) = \max\{x(v) \mid v \in S \wedge y(v) = \miY\}$. It
  follows from the definition that the neighbors of $\nx(v,\Gamma)$ (resp.,
  $\ny(v,\Gamma)$) are to the left of and above $\nx(v,\Gamma)$ (resp.,
  $\ny(v,\Gamma)$). Moreover, from the planarity of $\Gamma$, it follows that
  $\lab(\nx(v,\Gamma)) \in \{\R,\C\} \setminus \{\lab(v)\}$ and
  $\lab(\ny(v,\Gamma)) \in \{\R,\C\} \setminus \{\lab(v)\}$. See
  \cref{fig:nearXnY-4}.
\end{observation}

We now give a reduction rule similar to \cref{rer:RCR} for a path
labeled $\CRC$.

\begin{reduction}\label{rer:CRC}
  Suppose that there exists a path $P= (u_1, u_2, u_3,u_4,u_5)$ in
  $\dirG$ such that $\lab(u_2) = \C$, $\lab(u_3) = \R$ and
  $\lab(u_4) = \C$. Then, delete the vertex $u_3$ and the edges
  incident to it from $\dirG$ and add a new path $(u_2, u'_3, u_4)$ to
  $\dirG$. Moreover, assign $\lab(u_2) = \F$, $\lab(u'_3) = \C$,
  $\lab(u_4) = \F$, $\len((u_2,u'_3)) = \len((u_3,u_4))$ and
  $\len((u'_3,u_4)) = \len((u_2,u_3))$. The labels of all the
  remaining vertices and the weights of all the remaining edges stay
  the same. Let $\dirG_\redu$ be the reduced graph.
\end{reduction}

\begin{lemma}
  \label{lem:CRC}
  \cref{rer:CRC} is safe.
\end{lemma}

\begin{proof}
  The proof of the lemma follows similarly to the proof of
  \cref{lem:RCR}.
\end{proof}

We now give a lemma similar to \cref{lem:lenRCCR} for a path labeled $\CRRC$.

\begin{lemma}\label{lem:lenCRRC}
  Suppose that there exists a path $P = (u_1,u_2,u_3,u_4,u_5,u_6)$ in
  $\dirG$ such that $\lab(u_2) = \lab(u_5) = \C$ and
  $\lab(u_3) = \lab(u_4) = \R$. For any drawing $\Gamma$ of $\dirG$, there
  exists another drawing $\Gamma'$ of $\dirG$ whose bounding box is same as
  that of $\Gamma$ such that $y(u_2) = y(u_5)$ in $\Gamma'$.
\end{lemma}

\begin{proof}
  The proof of the lemma follows similarly to the proof of
  \cref{lem:lenRCCR}.
\end{proof}

We now give a lemma similar to \cref{lem:lenRCCC} for paths labeled $\CRRR$ or $\RRRC$.

\begin{lemma}\label{lem:lenCRRR}
  Suppose that there exists a path $P = (u_1,u_2,u_3,u_4,u_5,u_6)$ in
  $\dirG$ such that $\lab(u_2) = \C$ and
  $\lab(u_3) = \lab(u_4) = \lab(u_5) = \R$ (resp.,
  $\lab(u_2) = \lab(u_3) = \lab(u_4) = \R$ and $\lab(u_5) = \C$). For
  any drawing $\Gamma$ of $\dirG$, $y(u_2) < y(u_5)$ in $\Gamma$.
\end{lemma}

\begin{proof}
  The proof of the lemma follows similarly to the proof of
  \cref{lem:lenRCCC}.
\end{proof}

We now give reduction rules similar to \cref{rer:RCCC} for paths
labeled $\CCCR$, $\CRRR$ and $\RRRC$.

\begin{reduction}\label{rer:CCCR}
  Suppose that there exists a path $P= (u_1, u_2, u_3,u_4,u_5,u_6)$ in
  $\dirG$ such that $\lab(u_2) = \lab(u_3) = \lab(u_4) = \C$ and
  $\lab(u_5) = \R$. Then, delete the vertices $u_3$ and $u_4$ and the
  edges incident to them from $\dirG$ and add a new path
  $(u_2, u'_3, u_5)$ to $\dirG$. Moreover, assign $\lab(u_2) = \C$,
  $\lab(u'_3) = \C$, $\lab(u_5) = \F$,
  $\len((u'_3,u_5)) = \len((u_3,u_4))$,
  $\len((u_6, \need(u_6))) = \max\{$
  $\len((u_4,u_5)), \len((\need(u_1),$ $u_1))\}$, and
  $\len((u_2,u'_3)) = \max\{$ $\len((u_2,$
  $u_3)) - \len((u_4,u_5)), 1\}$. The labels of all the remaining
  vertices and the weights of all the remaining edges stay the
  same. Let $\dirG_\redu$ be the reduced graph.
\end{reduction}

\begin{lemma}
  \label{lem:CCCR}
  \cref{rer:CCCR} is safe.
\end{lemma}

\begin{proof}
  The proof of the lemma follows similarly to the proof of
  \cref{lem:RCCC}.
\end{proof}

\begin{reduction}\label{rer:CRRR}
  Suppose that there exists a path $P= (u_1, u_2, u_3,u_4,u_5,u_6)$ in
  $\dirG$ such that $\lab(u_2) = \C$ and
  $\lab(u_3) = \lab(u_4) = \lab(u_5) = \R$. Then, delete the vertices
  $u_3$ and $u_4$ and the edges incident to them from $\dirG$ and add
  a new path $(u_2, u'_3, u_5)$ to $\dirG$. Moreover, assign
  $\lab(u_2) = \F$, $\lab(u'_3) = \R$, $\lab(u_5) = \R$,
  $\len((u_2,u'_3)) = \len((u_3,u_4))$,
  $\len((\pred(u_1),u_1)) = \max\{$
  $\len((u_2,u_3)), \len((\pred(u_1),$ $u_1))\}$, and
  $\len((u'_3,u_5)) = \max\{$ $\len((u_4,$
  $u_5)) - \len((u_2,u_3)), 1\}$. The labels of all the remaining
  vertices and the weights of all the remaining edges stay the
  same. Let $\dirG_\redu$ be the reduced graph.
\end{reduction}

\begin{lemma}
  \label{lem:CRRR}
  \cref{rer:CRRR} is safe.
\end{lemma}

\begin{proof}
  The proof of the lemma follows similarly to the proof of
  \cref{lem:RCCC}.
\end{proof}

\begin{reduction}\label{rer:RRRC}
  Suppose that there exists a path $P= (u_1, u_2, u_3,u_4,u_5,u_6)$ in
  $\dirG$ such that $\lab(u_2) = \lab(u_3) = \lab(u_4) = \R$ and
  $\lab(u_5) = \C$. Then, delete the vertices $u_3$ and $u_4$ and the
  edges incident to them from $\dirG$ and add a new path
  $(u_2, u'_3, u_5)$ to $\dirG$. Moreover, assign $\lab(u_2) = \R$,
  $\lab(u'_3) = \R$, $\lab(u_5) = \F$,
  $\len((u'_3,u_5)) = \len((u_3,u_4))$,
  $\len((u_6, \need(u_6))) = \max\{$
  $\len((u_4,u_5)), \len((\need(u_1),$ $u_1))\}$, and
  $\len((u_2,u'_3)) = \max\{$ $\len((u_2,$
  $u_3)) - \len((u_4,u_5)), 1\}$. The labels of all the remaining
  vertices and the weights of all the remaining edges stay the
  same. Let $\dirG_\redu$ be the reduced graph.
\end{reduction}

\begin{lemma}
  \label{lem:RRRC}
  \cref{rer:RRRC} is safe.
\end{lemma}

\begin{proof}
  The proof of the lemma follows similarly to the proof of
  \cref{lem:RCCC}.
\end{proof}

We now give a reduction rule similar to \cref{rer:RCCR} for a path
labeled $\CRRC$.

\begin{reduction}\label{rer:CRRC}
  Suppose that there exists a path $P = (u_1,u_2,u_3,u_4,u_5,u_6)$ in
  $\dirG$ such that $\lab(u_2) = \lab(u_5) = \C$ and
  $\lab(u_3) = \lab(u_4) = \R$. Then, delete the vertices $u_3$ and
  $u_4$ and the edges incident to them from $\dirG$ and add a new edge
  $(u_2,u_5)$ to $\dirG$. Moreover, assign
  $\lab(u_2) = \lab(u_5) = \F$, $\len((u_2,u_5)) = \len((u_3,u_4))$
  and $\len((\pred(u_1),$ $u_1)) = \max\{\len((u_2,$
  $u_3)), \len((\pred(u_1),$ $u_1))\}$. The labels of all the
  remaining vertices and the weights of all the remaining edges stay
  the same. Let $\dirG_\redu$ be the reduced graph.
\end{reduction}

\begin{lemma}
  \label{lem:CRRC}
  \cref{rer:CRRC} is safe.
\end{lemma}

\begin{proof}
  The proof of the lemma follows the proof of \cref{lem:RCCR}.
\end{proof}

\end{document}